\pdfoutput=1
\RequirePackage{ifpdf}
\ifpdf 
\documentclass[pdftex]{sigma}
\else
\documentclass{sigma}
\fi

\numberwithin{equation}{section}

\newtheorem{Theorem}{Theorem}[section]

\newtheorem{Lemma}[Theorem]{Lemma}

{\theoremstyle{definition}
\newtheorem{Definition}[Theorem]{Definition}
\newtheorem{Example}[Theorem]{Example}
\newtheorem{Remark}[Theorem]{Remark}}

\def \wt{\widetilde}
\def \pa{\partial}
\def \ds{\displaystyle}
\def\nn{\nonumber}
\renewcommand{\d}{\mathrm d}
\newcommand{\C}{\mathbb{C}}
\def\p{\mathbf p}
\def\q{\mathbf q}
\renewcommand{\le}{\left}
\newcommand{\ri}{\right}

\newcommand{\1}{\mathbf{1}}

\def \Tr {\operatorname{Tr}}

\begin{document}
\allowdisplaybreaks

\newcommand{\arXivNumber}{2103.09681}

\renewcommand{\PaperNumber}{081}

\FirstPageHeading

\ShortArticleName{Quantization of Calogero--Painlev\'e System}

\ArticleName{Quantization of Calogero--Painlev\'e System and \\ Multi-Particle Quantum Painlev\'e Equations II--VI}

\Author{Fatane MOBASHERAMINI~$^{\rm a}$ and Marco BERTOLA~$^{\rm ab}$}
\AuthorNameForHeading{F.~Mobasheramini and M.~Bertola}

\Address{$^{\rm a)}$~Department of Mathematics and Statistics, Concordia University,\\
\hphantom{$^{\rm a)}$}~1455 de Maisonneuve W., Montreal, QC H3G 1M8, Canada}
\EmailD{\href{mailto:f_moba@live.concordia.ca}{f$\_$moba@live.concordia.ca}, \href{mailto:marco.bertola@concordia.ca}{marco.bertola@concordia.ca}}

\Address{$^{\rm b)}$~SISSA, Area of Mathematics, via Bonomea 265, Trieste, Italy}

\ArticleDates{Received March 19, 2021, in final form August 31, 2021; Published online September 07, 2021}

\Abstract{We consider the isomonodromic formulation of the Calogero--Painlev\'e multi-particle systems and proceed to their canonical quantization. We then proceed to the quantum Hamiltonian reduction on a special representation to radial variables, in analogy with the classical case and also with the theory of quantum Calogero equations. This quantized version is compared to the generalization of a result of Nagoya on integral representations of certain solutions of the quantum Painlev\'e equations. We also provide multi-particle ge\-neralizations of these integral representations.}


\Classification{70H08; 81R12}

\section{Introduction}\label{sec1}

The six Painlev\'e equations were discovered by \cite{Fuchs,Gam, Pain} as a result of the search for second-order ordinary differential equations in the complex plane, whose only movable singularities are poles, and are not generally solvable in terms of elementary functions. These equations were subject of extensive scrutiny by mathematicians and physicists, and they find a wide range of applications. These equations also arise as reductions of
soliton equations~\cite{Flas}, which are solvable by inverse scattering. Other
applications include quantum gravity and string
theory~\cite{Brez,Doug,Gros}, $\beta$-models~\cite{Boro}, topological field theories~\cite{Dub}, random matrices~\cite{ForWi, TW} and
stochastic growth processes~\cite{Lee}.

Historically, the Hamiltonian structure of the Painlev\'e equations was studied in
\cite{DubMaz,Mal,Oka} and they all can be written as a time-dependent Hamiltonian system
\begin{equation*}
\Ddot{q}=-V(q;t),
\end{equation*}
for some potential function $V$ of dependent and independent variables.

On the other hand, these equations admit a natural generalization to multi-particle Hamiltonian systems with an interaction of Calogero type (rational, trigonometric, or elliptic). This approach was pioneered by K.~Takasaki \cite{Tak} as an attempt at de-autonomization of a result by Inozemtsev~\cite{Inoz}. The result yields a new system of multi-component Hamiltonian operators that was named {\it Calogero--Painlev\'e correspondence.}

The integrability of this system was conjectured by Takasaki and was shown in \cite{BerCafRub} via an isomonodromic formulation in terms of a Lax pair of $2N\times 2N$ matrices, where $N$ is the number of particles.

On a parallel track, the integrability of the classical Calogero--Painlev\'e equations has attracted increasing interest about their quantization. As a result, the discussion about the integrability of such a system and their corresponding solutions has appeared in the work of mathematicians and mathematical physicists such as H.~Nagoya \cite{Nag}, K.~Okamoto \cite{Oka}, A.~Zabrodin and A.~Zotov~\cite{ZZ}.

In \cite{ZZ}, A.~Zabrodin and A.~Zotov show that all Painlev\'e equations {I}--{VI} can be presented in the form of the non-stationary Schr\"odinger equation in imaginary time
\begin{equation*}
\partial_t\Psi =H \Psi
\end{equation*}
with $H$ the standard Schr\"odinger operator
\begin{equation}
H=\frac{1}{2}\partial_z^2+V(z,t),\label{Schro}
\end{equation}
which is a natural quantization of the classical Calogero--Painlev\'e Hamiltonian operator
\begin{equation*}
H(p,q,t)=\frac{p^2}{2}+V(q;t).
\end{equation*}
In this article, we follow a similar logic starting from the Hamiltonian formulation of the multi-particle Hamiltonians in \cite{BerCafRub}, with a different approach to the system that is called the quantum Calogero--Painlev\'e Hamiltonian system. This task is achieved by applying the canonical quantization of the form (described in details in Section~\ref{sec2})
\begin{gather*}
q_{ij} \longrightarrow q_{ij} \qquad \text{and} \qquad p_{ij} \longrightarrow \hbar\frac{\partial}{\partial q_{ji}}
\end{gather*}
to the Hamiltonian operators in isomonodromic formulation introduced in \cite{BerCafRub}. As a result of this quantization, we obtain a multi-particle quantized Hamiltonian system that satisfies the Schr\"odinger equation of the same form as~\eqref{Schro}.

We also compare the result of this quantization to what H.~Nagoya \cite{Nag} introduces as quantum Painlev\'e Hamiltonian system for a single particle wave function with coordinate $z$.
We generalize these Hamiltonian operators to the case of $N$ particles with coordinates $z_\rho$, $\rho=1,\dots ,N$.

We extend the integral representation of solutions given in~\cite{Nag} to the multi-component integral solution of the Schr\"odiner equation for $N$-particle quantum Painlev\'e equations. These integral representations are defined only for the quantum Painlev\'e equation II--VI; these solutions should be understood as the quantum counterpart of rational solutions and it is well known that the first Painlev\'e equation does not admit rational solutions, which heuristically explains the absence of integral representations for solutions thereof. Intriguingly, these integral representations are presented as some type of $\beta$ integrals that appear in other areas such as conformal field theory~\cite{JimNag}, $\beta$-ensembles~\cite{Burger}, the theory of orthogonal polynomials, and hypergeometric functions.

Finally, we show that under some constraints on the parameters of the Hamiltonian operators obtained in the first two sections of this article these generalized Nagoya integrals provide solutions for the quantization of the multi-particle Hamiltonian systems described in the first part.

\section{Isomonodromic formulation and quantization}\label{sec2}

In \cite{BerCafRub} the authors provided the answer to a conjecture proposed by K.~Takasaki in~\cite{Tak}; Takasaki considered the de-autonomization of the Calogero systems proposed by Inozemtsev~\cite{Inoz} by observing that for Painlev\'e VI the rank-one Inozemtsev system reduces to the autonomous version of the Hamiltonian form of Painlev\'e VI that was first written by Gambier and then rediscovered by Manin. This leads to the postulation of the deautonomized versions, and the conjecture that they should be describing isomonodromic deformations of an appropriate system. For this reason, Takasaki coined the term ``Painlev\'e--Calogero correspondence''.

The core idea of \cite{BerCafRub} is as follows; they consider a (complexified) phase space consisting of pairs of $N\times N$ matrices $\p$, $\q$, identified with the cotangent bundle of $X = {\rm Mat}_{N\times N}(\mathbb C)$. The (complex) symplectic form is then
\begin{equation}
\label{sympstr}
\omega = \Tr ( \d \p \wedge \d \q) = \sum_{i,j} \d p_{ij} \wedge \d q_{ji} \quad \Leftrightarrow\quad \{p_{ab}, q_{cd}\} = \delta_{ad} \delta_{bc}.
\end{equation}
Consider an Hamiltonian $H(\p,\q)$ which is conjugation invariant, namely
\begin{gather*}
H(\p,\q) = H\big(C\p C^{-1} ,C\q C^{-1}\big) \qquad \hbox{with} \quad C\in {\rm GL}_N(\C).
\end{gather*}
Then Noether's theorem guarantees the conservation of the associated momenta $M = [\p,\q]$ for this group hamiltonian action. This allows us to fix a particular value of the momentum $M$ and investigate the reduced system on the leaf of this value.
If the momentum $M$ is fixed to be of the form
\begin{gather}\label{mom}
M=[\p,\q] = {\rm i}g\big(\1-v^{\rm T} v\big),\qquad \hbox{with} \quad v:= (1,\dots, 1),
\end{gather}
then one can apply a theorem used in the classical theory of Calogero system and due to Kazhdan, Kostant, and Sternberg \cite{KKS}. It states that we can diagonalize $\q = CXC^{-1}$ with $X= \operatorname{diag} (x_1,\dots, x_n)$ in such a way that the matrix $Y = C^{-1} \p C$ is of the form
\begin{gather*}
Y = \operatorname{diag} (y_1,\dots, y_N) + \le[\frac {{\rm i}g}{x_j-x_k}\ri]_{j,k=1}^N.
\end{gather*}
The variables $y_j$'s are the momenta conjugated to the eigenvalues $x_j$ in the reduced system.
Then the first result of~\cite{BerCafRub} was that all the Calogero--Painlev\'e systems of~\cite{Tak} are the reduction on the particular value of the momentum~\eqref{mom} of a list of conjugation-invariant Hamiltonian systems:
\begin{gather}
\wt{H}_{\rm I}
= \Tr\left(\frac{\p^2}{2}-\frac{\q^3}{2}-\frac{t\q}{4}\right),\nn\\
\wt{H}_{\rm II}
=\Tr\left(\frac{\p^2}{2}-\frac{1}{2}\left(\q^2+\frac{t}{2}\right)^2-\theta \q\right),\nn \\
t\wt{H}_{\rm III}
= \Tr\left(\p^2\q^2 -\big(\q^2+(\theta_0-\theta_1)\q-t\big)\p-\theta_1 \q\right),
\nn \\
\wt{H}_{\rm IV}
=\Tr\left(\p\q (\p-\q-t )+\theta_0 \p-(\theta_0+\theta_1)\q\right),\nn \\
t\wt{H}_{\rm V}
=\Tr\left(\p(\p+t)\q(\q-1)+(\theta_0-\theta_2)\p\q+\theta_2 \p+(\theta_0+\theta_1)t\q\right),\nn \\
t(t-1)\wt{H}_{\rm VI}
=\Tr\bigg(\q\p\q\p\q-t\p\q^2\p+t\p\q\p-\p\q\p\q-\theta \q\p\q
+t(\theta_0+\theta_1)\p\q
\nn \\
\hphantom{t(t-1)\wt{H}_{\rm VI}=}{} +(\theta_0+\theta_t)\p\q-\theta_0 t\p
 -\frac{1}{4}\big(k^2-\theta^2\big)\q\bigg),\label{IsoHamil}
\end{gather}
where $\theta$, $\theta_0$, $\theta_1$, $\theta_2$, $\theta_t$, and $k$ are arbitrary parameters in~$\mathbb{C}$.
These Hamiltonians should be thought of as non-commutative polynomials in $\p$, $\q$ generalizing the Okamoto Hamiltonians for the six Painlev\'e equations.

They showed that these Hamiltonians describe the isomonodromic deformations of an ODE in the $z$-plane for a matrix $\Phi(z)$ of size $2N\times 2N$, which reduces, for $N=1$, to the classical Lax pair formulation for Painlev\'e equations (see, e.g.,~\cite{JMU}).

To give a meaningful nontrivial example, we consider the Painlev\'e VI case:
\begin{equation*}
 \begin{cases}
	 \ds \frac{\partial \Phi}{\partial z} = \left( \frac{A_0}{z} + \frac{A_1}{z - 1} + \frac{A_t}{z-t} \right)\Phi= {\bf A}(z) \Phi,\vspace{1mm}\\
\ds \frac{\partial \Phi}{\partial t} = -\left(\frac{A_t}{z-t} + B\right)\Phi = {\bf B}(z) \Phi,
\end{cases}
\end{equation*}
where the matrices are explicitly given by
\begin{gather*}
A_0 := \left[ \begin{matrix} -1-\theta_{{t}}&\ds\frac{\q }{t} -1
\\ 0&0\end{matrix} \right],
\qquad
A_1 := \left[ \begin{matrix} -\q \p + \ds\frac{1}2(k + \theta) & 1 \vspace{1mm}\\
				 (\theta - \q \p)\q \p + \ds\frac{1}4\big(k^2 - \theta^2\big) & \q \p + \ds\frac{1}2(k - \theta)
	\end{matrix} \right], \nn
\\	
		A_t := \left[ \begin{matrix} \q \p -\theta_0 & -\ds\frac{\q }t \\
				t(-\theta_0 + \p\q )\p	 & -\p\q \end{matrix} \right], \qquad
	B := \left[ \begin{matrix} \ds\frac{t([\q ,\p]_+ -\theta_0) + \theta \q - [\q \p,\q ]_+}{t(t-1)} & 0 \\
					-\theta_0 \p + \p\q \p & 0 \end{matrix} \right],
\end{gather*}
where $\theta=\theta_0+\theta_1+\theta_t$. The expression $[X,Y]_+$ stands here for the anti-commutator of the noncommutative symbols $X$, $Y$, namely $[X,Y]_+ = XY+YX$.
The partitioning is in $N\times N$ block, and all scalars are automatically considered mutliple of the identity matrix of size $N$.
The isomonodromic equations consist in the ``zero-curvature'' equations for the pair
\begin{gather*}
\frac {\pa {\bf A}(z)}{\pa t} -\frac {\pa {\bf B}(z)}{\pa z} +[{\bf A}(z), {\bf B}(z)]=0,
\end{gather*}
and, with some elementary algebra, they become the following evolutionary system for the operators $\p$, $\q$:
 	\begin{gather}\label{dynamicsPVI}
	 \begin{cases}
		\dot{\q } = \mathcal A(\q ,\p), \\
		\dot{\p} = \mathcal B(\q ,\p),
	\end{cases}
	\end{gather}
where the non-commutative polynomials $\mathcal A$, $\mathcal B$ are given by
	\begin{gather}
	t(t-1)\mathcal A(\q ,\p) := -\theta_0 t + (\theta_0 + \theta_t)\q + (\theta_0 + \theta_1)t\q -\theta \q ^2 -2 \q \p\q \nonumber \\
\hphantom{t(t-1)\mathcal A(\q ,\p) :=}{} + t[\p,\q ]_+ - \big[t\p,\q ^2\big]_+ + [\q \p\q ,\q ]_+, \nonumber \\
	t(t-1)\mathcal B(\q ,\p) := \frac{1}4\big(k^2 - \theta^2\big) -(\theta_0 + \theta_t)\p - (\theta_0 + \theta_1)t\p + \theta[\q ,\p]_+ - t\p^2 \nonumber\\
\hphantom{t(t-1)\mathcal B(\q ,\p) :=}{} + t\big[\q ,\p^2\big]_+ + \p\big(2\q - \q ^2\big)\p - [\q ,\p\q \p]_+ .\label{defABPVI}
\end{gather}
The key observation, which is the initial thrust of our present paper, is the following: in the above derivation of the zero curvature equations the symbols $\p$, $\q$ can be taken in some arbitrary non-commutative algebra. In the case of \cite{BerCafRub} where~$\p$, $\q$ are matrices, the above equations turn out, by inspection, to be Hamiltonian equations with respect to the symplectic structure~\eqref{sympstr} with an Hamiltonian given by
 \begin{gather*}
t(t-1)H =
\Tr \bigg( \q\p\q\p\q - t \p\q^2 \p + t \p\q\p - \p\q\p\q - \theta\q\p\q + t(\theta_0 + \theta_1) \p\q\\
\hphantom{t(t-1)H =}{} + (\theta_0 + \theta_t)\p\q - \theta_0 t \p - \frac{1}{4} \big(k^2- \theta^2\big) \q\bigg).
\end{gather*}
Analogous considerations apply to each of the other cases, with the Hamiltonians given in \eqref{IsoHamil}.

{\bf Quantization.}
In view of the considerations above, we want to consider the canonical quantization of the symplectic structure \eqref{sympstr}. The main logic is that we keep equations \eqref{dynamicsPVI} (and the other equations corresponding to each of the other cases listed in~\cite{BerCafRub}) and seek a Hamiltonian formulation with {\it quantum} Hamiltonians. The canonical quantization in ``Schr\"odinger'' representation amounts to considering the entries of $\q$ as multiplication operators and the entries of~$\p$ corresponding differential operators as follows
\begin{equation}
q_{ij} \longrightarrow q_{ij} \qquad \text{and} \qquad p_{ij} \longrightarrow \hbar\frac{\partial}{\partial q_{ji}}.
\label{quan}
\end{equation}

The effect of this canonical quantization is that we cannot simply take the expressions~\eqref{IsoHamil} as Hamiltonians generating the relevant equations of motions like \eqref{dynamicsPVI} because there are issues of normal ordering. To explain the issue we point out that in the classical case the expressions $\Tr (\p\q)$ and $\Tr(\q\p)$ coincide, but if $\p$,~$\q$ are quantum operators~\eqref{quan} then these two expressions differ. This should explain why the {\it quantum} version of the Hamiltonians~\eqref{IsoHamil} will be slightly different due to the fact that the correct scheme depends on the non-commutativity of the traces in this case. Note that the commutation relations that lead to these Hamiltonian operators read the following equations for each of the Calogero--Painlev\'e equations:
\begin{gather*}
\hbar \dot{\q}=[H_J,\q] \qquad \text{and} \qquad \hbar \dot{\p}=[H_J,\p], \qquad J \in \{{\rm I},\dots, {\rm VI}\}.
\end{gather*}
The correct Hamiltonian operators to which we apply the quantization, are the following ope\-rators
\begin{gather}
t\wt H_{\rm III}= \Tr\bigg(\frac{\p^2\q^2+\q^2\p^2}{2}-\frac{\q^2\p+\p\q^2}{2} -(\theta_0-\theta_1)\q\p+t\p-\theta_1\q\bigg),
\nonumber\\
\wt H_{\rm IV} =\Tr\bigg(\p\q\p-\frac{\p\q^2+\q^2\p}{2}-t\p\q+\theta_0p-(\theta_0+\theta_1)\q\bigg),\nonumber\\
t\wt H_{\rm V}=\Tr\bigg(\frac{\p^2\q^2+\q^2\p^2}{2}-\frac{\p^2\q+\q\p^2}{2}+\frac{t\big(\p\q^2+\q^2\p\big)}{2} \nonumber\\
\hphantom{t\wt H_{\rm V}=}{}
+(\theta_0-\theta_2-t)\p\q+\theta_2\p+(\theta_0+\theta_1)t\q\bigg),\nonumber\\
t(t-1) \wt H_{\rm VI}
=\Tr\bigg(\q\p\q\p\q-t\p\q^2\p+t\p\q\p-\frac{\p\q\p\q+\q\p\q\p}{2}-\theta \q\p\q
\nonumber\\
\hphantom{t(t-1) \wt H_{\rm VI}=}{}+t(\theta_0+\theta_1)\p\q+(\theta_0+\theta_t)\p\q-\theta_0t\p-\frac{1}{4}\big(k^2-\theta^2\big)\q\bigg),\label{revisedCP}
\end{gather}
whereas the Painlev\'e I and Painlev\'e II Hamiltonians remain formally the same. Note that we use $\p$,~$\q$ here and below to denote the {\it quantum} operators, without further notice.

The operators above are determined, by inspection, requiring that the evolution equations for $\p$,~$\q$ remain of the same form as those in \cite{BerCafRub}.

\begin{Example}We compute $[\p,\Tr(\p\q\p\q)]$ both quantistically and classically to show the origin of the difference in the quantum Hamiltonians~\eqref{revisedCP}. We start with the classical computation where $\{p_{ij}, q_{k\ell}\}= \delta_{i\ell} \delta_{jk}$:
\[
\big\{\p,\Tr(\p\q\p\q)\big\} = 2 \p\q\p.
\]
In computing this we have used also the cyclicity of the trace.

Viceversa, considering the quantized operators, the same
expression yields
\begin{gather*}
[\p,\Tr(\p\q\p\q)]_{ij} =
\sum_{\alpha,\beta,\mu,\nu} \le[p_{ij},p_{\alpha\beta}q_{\beta \mu}p_{\mu \nu}q_{\nu \alpha}\ri]
\nonumber \\
\hphantom{[\p,\Tr(\p\q\p\q)]_{ij}}{}
= \sum_{\alpha,\beta,\mu,\nu} p_{\alpha\beta}\le(p_{ij}q_{\beta \mu}p_{\mu \nu}q_{\nu \alpha}-q_{\beta \mu}p_{\mu \nu}q_{\nu \alpha}p_{ij}\ri)
\end{gather*}
we add and subtract $q_{\beta \mu}p_{ij}p_{\mu \nu}q_{\nu \alpha}$ to the expression inside the bracket, combining the terms and using the commutation relation $[p_{ij},q_{kl}]=\hbar \delta_{il}\delta_{jk}$:
\begin{gather*}
[\p,\Tr(\p\q\p\q)]_{ij}= \hbar \sum_{\alpha,\beta,\mu,\nu} p_{\alpha\beta}\le( \delta_{i\mu}\delta_{j\beta}p_{\mu\nu}q_{\nu\alpha}+ \delta_{i\alpha} \delta_{j\nu}q_{\beta\mu}p_{\mu\nu}\ri) \nonumber \\
\hphantom{[\p,\Tr(\p\q\p\q)]_{ij}}{} = \big( \hbar \p^2\q+ \hbar\p\q\p \big)_{ij}= \big(2\hbar \p\q\p + \hbar^2 \p\big)_{ij}.
\end{gather*}

 Since the desired term in the equation of motion is $2\hbar \p\q\p$ we need to replace the term $\Tr(\p\q\p\q)$ in the Hamiltonian with a ``symmetrized'' version $\frac{1}{2}\Tr(\p\q\p\q+\q\p\q\p)$. Indeed, one then similarly computes:
\begin{gather*}
\le[\p,\frac{1}{2}\Tr(\p\q\p\q+\q\p\q\p)\ri] = 2 \hbar \p\q\p,
\end{gather*}
which is one of the terms that appears in the expression of $\mathcal A$ in equation of motion~\eqref{defABPVI}.
\end{Example}

\subsection{Quantization in radial form}
Having established the correct context of the quantization of the isomonodromic equations we proceed now to the quantum version of the Kazhdan--Kostant--Sternberg reduction \eqref{mom}. The first step is to express the quantum Hamiltonians~\eqref{revisedCP} in terms of the eigenvalues. The fact that this is at all possible is simply a consequence of the invariance of the Hamiltonians under conjugations.
To do so, we need to use the Harish-Chandra homomorphism \cite{Etin}.

\begin{Definition}
Let $\mathcal M$ be the manifold of diagonalizable matrices and denote
\begin{equation*}
\mathcal{D}(\mathcal M)^{\rm G}\qquad \text{and} \qquad \mathcal{D}({\rm Diag}({\rm GL}_n))^{\rm W}
\end{equation*}
the adjoint-invariant subalgebra of the algebra of differential operators over $\mathcal M$ and the Weyl-invariant subalgebra of the algebra of differential operators over diagonal matrices, respectively.

The canonical isomorphism
\begin{equation}
\mathcal{H}_{\rm c}\colon \ \mathcal{D}(M)^{\rm G}\longrightarrow \mathcal{D}({\rm Diag}({\rm GL}_n))^{\rm W}\label{HCmap}
\end{equation} is called the Harish-Chandra map.
\end{Definition}

This means the following: for a character function $\Psi(Q)$ (i.e., $\Psi(Q)=\Psi\big(GQG^{-1}\big)$, $G\in {\rm GL}_n$) and $\mathcal{L}$ a differential operator invariant under the adjoint map we have
\begin{equation*}
(\mathcal{L}\Psi)|_{\rm Diag}=\mathcal{H}_{\rm c}(\mathcal{L})(\Psi|_{\rm Diag}),
\end{equation*}
where $\mathcal{H}_{\rm c}(\mathcal{L})$ is a differential operator on the eigenvalues.

Our goal now is to make this isomorphism completely explicit and subsequently express all Hamiltonians in~\eqref{revisedCP} as differential operators acting on the eigenvalues when applied to character functions or pseudo-character functions, namely $\Psi(\q)= \Psi\big(G\q G^{-1}\big) {\rm e}^{\theta(G,\q)}$, with $\theta$ an appropriate cocycle.

{\bf Explicit construction of Harish-Chandra isomorphism.}
To make the Harish-Chandra homomorphism \eqref{HCmap} explicit, we write the matrix $Q=Z+M$ with $Z$ diagonal and $M$ off-diagonal: we then act with an infinitesimal conjugation up to order two in $M$ to diagonalize it.
Concretely this means the following; we conjugate the matrix $Q= Z + M$ by a matrix of the form $G:= {\rm e}^{A^{(1)} + A^{(2)}}$ where $A^{(1)}$ is assumed to be of first order in the entries of $M$ and $A^{(2)}$ of second order and both are off-diagonal matrices.

We then impose that the conjugation of $Q$ by $G$ is diagonal up to order $2$.
Then
\begin{gather}
{\rm e}^{{\rm ad}_{A^{(1)}+A^{(2)}}}(Z+M) =
Z+M+\big[A^{(1)}+A^{(2)},Z+M\big]+\frac{1}{2}\big[A^{(1)},\big[A^{(1)},Z\big]\big]+\mathcal{O}(3) \nn \\
 \qquad{} =
Z+M+\big[A^{(1)},Z\big]+\big[A^{(1)},M\big]+\big[A^{(2)},Z\big]+\frac{1}{2}\big[A^{(1)},\big[A^{(1)},Z\big]\big]+\mathcal{O}(3),\label{Zad}
\end{gather}
where $\mathcal O(3)$ denotes terms of order $3$ or higher in the entries of $M$. We need to impose that the result is a diagonal matrix up to the indicated order. Separating the equations according to their order in $M$ we obtain
\begin{gather}
\big[Z,A^{(1)}\big] =M \qquad \text{at order 1},\label{A1}\\
\big[Z,A^{(2)}\big] = \big[A^{(1)},M\big] + \frac 12 \big[A^{(1)},\big[A^{(1)},Z\big]\big] \qquad \text{at order 2}.\label{A2}
\end{gather}
The matrices $A^{(1,2)}$ are off-diagonal, and the ${\rm ad}_Z$ operator is invertible on the subspace of off-diagonal matrices so that we can solve the two equations above to obtain
\begin{gather}
A^{(1)}_{ab}=\frac{M_{ab}}{z_a-z_b}, \nonumber\\ A^{(2)}_{ab}=-\frac{1}{2}\frac{\big[A^{(1)},\big[A^{(1)},Z\big]\big]_{ab}}{z_a-z_b}=\frac{1}{2}\frac{\big[A^{(1)},M\big]_{ab}}{z_a-z_b} =\frac{M_{ac}M_{cb}}{(z_a-z_c)(z_a-z_b)}.\label{A12}
\end{gather}

Substituting \eqref{A12} into \eqref{Zad} we obtain a diagonal matrix $\wt Z$ which is a shift of the matrix $Z$ as follows
\begin{gather}
\wt{Z}= Z+\frac{1}{2}\big[A^{(1)},M\big]_{D}+\mathcal{O}(3)\nonumber\\
\hphantom{\wt{Z}}{} =A^{(2)}+\frac{1}{2}\operatorname{diag}\left(\sum_{c}\frac{M_{\bullet c}}{z_\bullet -z_c}M_{c\bullet}-M_{\bullet c}\frac{M_{c\bullet}}{z_c-z_\bullet}\right)_{\bullet=1}^{N}+\mathcal{O}(3)\nonumber \\
\hphantom{\wt{Z}}{} =
Z+\operatorname{diag}\left(\sum_{d}\frac{M_{\bullet d} M_{d\bullet}}{z_\bullet- z_d}\right)_{\bullet=1}^N.\label{diag2}
\end{gather}
We now show how to use the above diagonalization to second order \eqref{diag2} to construct the Harish-Chandra homomorphism; we anticipate that the reason why we expand up to order $2$ is that all the operators we consider are at most quadratic in the momenta $\p$ and hence translate to differential operators of order $2$. If we had to consider the Harish-Chandra homomorphism for operators of the higher order, we would have to perform the above diagonalization up to the corresponding order.

{\bf Space of radial functions.}
When considering the quantum version of the Kazhdan--Kostant--Sternberg reduction, the choice of the special value of the momentum $M$ \eqref{mom} is replaced by the requirement that the quantum operators act on specific representations. We start with a general discussion on equivariant functions.

Let $\mathbb V$ be a vector space carrying a representation $\gamma$ of $G = {\rm GL}_N$ and let $\Psi\colon {\rm Mat}_{N\times N}(\C)= \operatorname{Lie}(G) \to \mathbb V$ an {\it $\gamma$-equivariant} function in the sense that
\begin{gather*}
\Psi(Q) = \gamma\big(g^{-1}\big) \Psi\big(gQ g^{-1}\big),
\end{gather*}
where $g\in {\rm GL}_ N(\C)$ and $\gamma\colon {\rm GL}_N(\C
) \to \operatorname{Aut}(\mathbb V)$ is a representation. Let us denote
\begin{gather*}
H_\gamma:= \big\{\Psi\colon {\rm Mat}_{N\times N}(\C) \to V, \, \text{$\gamma$-equivariant}\big\}.
\end{gather*}
For \looseness=-1 simplicity, we denote by the same symbol $\gamma$ the representation of ${\rm SL}_n$, the corresponding representation of the Lie algebra $\mathfrak g= \mathfrak {sl}_N$ as well as its natural extension to the universal enveloping algebra~$U(\mathfrak g)$.
We recall that the zero-weight subspace of the ${\rm SL}_N$-representation~$\mathbb V$ is
\begin{gather*}
\mathbb V(0) := \bigcap_{H\in \mathfrak h} \operatorname{Ker} \gamma(H),
\end{gather*}
where $\mathfrak h$ is the Cartan subalgebra of $\mathfrak{sl}_N$ (traceless diagonal matrices).

\begin{Lemma}\label{lemmazero}
Let $\mathcal D\subset {\rm Mat}_{N\times N}(\C)$ consist of the subspace of diagonal matrices. Then any equivariant function $\Psi$ restricts to a function from $\mathcal D$ to $\mathbb V(0)$.
\end{Lemma}
\begin{proof}
Consider a matrix $g ={\rm e}^{\epsilon H}$ with $H\in \mathfrak h$; then $\Psi(Q) = \gamma\big({\rm e}^{\epsilon H}\big) \Psi\big({\rm e}^{\epsilon H} Q{\rm e}^{-\epsilon H}\big)$. Restricting $Q=Z\in \mathcal D$ we have $\Psi(Z) = \gamma\big({\rm e}^{\epsilon H}\big) \Psi( Z)$. We now take the derivative with respect to $\epsilon$ at $\epsilon=0$ and we obtain $\gamma(H) \Psi(Z)=0$. Since~$H$ is arbitrary in $\mathfrak h$ it follows that $\Psi(Z) \in \mathbb V(0)$.
\end{proof}

Following \cite{Etin} the quantum Hamiltonian reduction that corresponds to the Kazhdan--Kostant--Sternberg orbit consists in taking a particular representation $\gamma$ of $\mathfrak {sl}_N$; the main feature of the $\mathfrak g$-module (which we denote by $\mathbb V_\kappa$) is that the zero weight space $\mathbb V_\kappa(0)$ is unidimensional. We denote with ${\mathbf c}$ a spanning element.
Specifically, $\mathbb V_\kappa$ consists of the space of functions of the form
\[
F(\xi_1,\dots, \xi_N)= \left(\prod_{j=1}^N \xi_j\right)^\kappa f\big( \vec \xi\,\big),
\]
where $f\big(\vec \xi\,\big)$ is a rational function with zero degree of homogeneity. The representation of the Lie algebra $\mathfrak {sl}_N$ is then the one obtained by restriction of the following $\mathfrak {gl}_N$ representation
\begin{gather}
\label{Vk}
\gamma(\mathbb E_{ab}) = \xi_a \frac {\pa}{\pa \xi_b}, \qquad a, b=1,\dots, N.
\end{gather}
It is easy then to see that $\mathbb V_\kappa(0) = \C\big\{\prod_{j=1}^N \xi_j^\kappa\big\}$.

Keeping this in mind we illustrate the type of computations needed to compute the extended Harish-Chandra homomorphism in the following example.

\begin{Example}
To illustrate the type of computations necessary, we consider the quantum radial reduction of the operator $\Tr\big(\q^k\p^2\big)$.
Using the form of the quantum operators $\p$, $\q$ we obtain
\[ 
\operatorname{Tr}\big(\q^k\p^2\big)\Psi(Q)
=
\bigg(\sum_{\rho,\sigma,\tau}\big(q^k\big)_{\rho \sigma}p_{\sigma \tau}p_{\tau \rho}\bigg)\Psi(Q)=
\bigg(\hbar^2 \sum_{\rho,\sigma,\tau}\big(q^k\big)_{\rho \sigma}\partial_{q_{\tau\sigma}}\partial_{q_{\rho\tau}}
\bigg)\Psi(Q).
\]
Since the function $\Psi$ is $\gamma$-equivariant and the operator is ${\mathrm{Ad}}$-invariant we can write $\Psi(Q) = \gamma(g)\Psi(Z)$ where $g$ is the matrix diagonalizing $Q$, and $Z$ is the diagonal matrix of its eigenvalues (this can be done on the set of diagonalizable matrices $Q$ whose complement of non-diagonalizable matrices is of zero measure and hence inessential to our considerations). We then consider matrices of the form $Q = Z+M$ with $Z$ diagonal and $M$ off-diagonal and its diagonalization up to order $2$ as in \eqref{Zad}. We then need to perform the derivatives and, at the end of the computation, restrict to the locus of diagonal matrices $Q=Z$. Using equation \eqref{diag2} and the matrices $A^{(1,2)}$ introduced in \eqref{A1} and \eqref{A2} we can continue the above computation by noticing that the terms involving the multiplication operator $\q$ can be directly evaluated at $Q= Z$ setting $M=0$:
\begin{gather}
\operatorname{Tr}\big(\q^k\p^2 \big)\Psi(Q)=
\le(\hbar^2 \sum_{\rho,\sigma,\tau}\delta_{\rho \sigma}{z}^k_\sigma \partial_{q_{\tau\sigma}}\partial_{q_{\rho\tau}}
\ri)\gamma\big({\rm e}^{-A^{(1)}-A^{(2)}}\big)\Psi\big(\tilde{Z}\big)=
\nn \\
\hphantom{\operatorname{Tr}\big(\q^k\p^2 \big)\Psi(Q)}{} =
\le(\hbar^2 \sum_{\sigma,\tau}z^k_\sigma \partial_{q_{\tau\sigma}}\partial_{q_{\sigma\tau}}
\ri)\gamma\big({\rm e}^{-A^{(1)}-A^{(2)}}\big)\Psi\big(\tilde{Z}\big).
\label{223}
\end{gather}
The second-order operator $\sum_{\rho,\sigma}z^k_\sigma \pa_{q_{\sigma\rho}} \pa_{q_{\rho\sigma}}$ written in terms of $Z$, $M$ becomes the operator $ \sum_{\rho} z_\rho^k\pa_{z_\rho}^2 + \sum_{\substack{\rho,\sigma \\ \rho\neq \sigma}}z^k_\rho \partial_{M_{\sigma \rho}}\partial_{M_{\rho \sigma}}$; the part involving the derivatives with respect to $z_\rho$ can be directly evaluated at $Q=Z$ while we postpone the evaluation of the part involving the derivatives in~$M_{\rho,\sigma}$:
\begin{gather}
\eqref{223}=
\hbar^2 \sum_{\sigma}z^k_\sigma \partial_{z_{\sigma}}^2 \Psi(Z)
+
\hbar^2 \underbrace{\sum_{\substack{\sigma ,\tau \\ \sigma\neq \tau}}z^k_\sigma \partial_{M_{\tau\sigma}}\partial_{M_{\sigma \tau}}\gamma\big({\rm e}^{-A^{(1)}-A^{(2)}}\big)\Psi\big(\tilde{Z}\big)}_{*},
\label{quan2}
\end{gather}
where $\wt Z = Z+\operatorname{diag}\Big(\sum_{d}\frac{M_{\bullet d} M_{d\bullet}}{z_\bullet- z_d}\Big)_{\bullet=1}^N
 $ as in~\eqref{diag2}.
Consider now the term marked with an asterisk: since $\wt Z-Z$ is a quadratic expression in the entries of~$M$, if we differentiate {\it} once~$\Psi$ or~$\gamma$ by $M_{\rho\sigma}$, by the chain rule there will be a multiplication by entire of~$M$ in the result. Thus, subsequent evaluation at $M=0$ will eliminate these terms. Therefore we need to consider the second-order operator acting on~$\gamma$ or~$\Psi$ separately. When acting on~$\Psi$ we have
\begin{gather*}
\sum_{\substack{\rho,\sigma \\ \rho\neq \sigma}}z^k_\sigma\partial_{M_{\sigma \rho}}\partial_{M_{\rho \sigma}} \Psi\big(\wt{Z}\big)\bigg|_{M=0}
=
\sum_{\substack{\rho,\sigma \\ \rho\neq \sigma}}z^k_\sigma\frac \pa{\partial{M_{\sigma \rho}}}\le( \frac {M_{\sigma \rho} \pa_{z_{\sigma} }\Psi}{z_\sigma - z_\rho} +\frac { M_{\sigma\rho} \pa_{z_\rho}\Psi}{z_\rho-z_\sigma}\ri) \bigg|_{M=0} \\
\hphantom{\sum_{\substack{\rho,\sigma \\ \rho\neq \sigma}}z^k_\sigma\partial_{M_{\sigma \rho}}\partial_{M_{\rho \sigma}} \Psi\big(\wt{Z}\big)\bigg|_{M=0}}{} =
\sum_{\substack{\rho,\sigma \\ \rho\neq \sigma}}z^k_\sigma\le(\frac {\pa_{z_\sigma} -\pa_{z_\rho}}{z_\sigma-z_\rho}\ri)\Psi(Z).
\end{gather*}

For the computation of the second term involving the representation $\gamma$, we note that
\begin{gather*}
\partial_{M_{\sigma \tau}}A^{(1)} = \sum_{\substack{\sigma ,\tau \\ \sigma\neq \tau}}\frac{E_{\sigma \tau}}{z_\sigma-z_\tau} , \quad \partial_{M_{\tau\sigma}}\partial_{M_{\sigma \tau}}A^{(1)}=0,
 \\
\partial_{M_{\sigma \tau}}A^{(2)}= \sum_{\substack{\sigma ,\tau, \nu \\ \nu\neq \sigma\neq \tau}} \frac{M_{\tau \nu}E_{\sigma\nu}}{(z_\sigma-z_\tau)(z_\sigma-z_\nu)}+\sum_{\substack{\sigma ,\tau, \mu \\ \mu\neq \sigma\neq\tau}}\frac{M_{\mu\sigma}E_{\mu\tau}}{(z_\mu-z_\sigma)(z_\mu-z_\tau)} , \qquad \partial_{M_{\tau\sigma}}\partial_{M_{\sigma \tau}}A^{(2)}=0.
\end{gather*}
Therefore the action of these differential operators on the group element ${\rm e}^{-A^{(1)}-A^{(2)}}$ gives (retaining the terms up to order $2$ in $M$ in the expansion of the exponential, since all higher order terms will give zero contribution when evaluated at $M=0$)
\begin{gather*}
\partial_{M_{\sigma \tau}}{\rm e}^{-A^{(1)}-A^{(2)}} = -\partial_{M_{\sigma \tau}}A^{(1)}-\partial_{M_{\sigma \tau}}A^{(2)}+\frac{\big[\partial_{M_{\sigma \tau}}A^{(1)},A^{(1)}\big]_+}{2}, \\
\partial_{M_{\tau\sigma}}\partial_{M_{\sigma \tau}}{\rm e}^{-A^{(1)}-A^{(2)}}\bigg |_{M=0}=-\frac{[\mathbb E_{\tau\sigma},\mathbb E_{\sigma \tau}]_+}{2(z_\sigma-z_\tau)^2}.
\end{gather*}
Here $\mathbb E_{\tau\sigma}$ denote the elementary matrices. Note also that $\mathbb E_{\tau\sigma}\mathbb E_{\sigma\tau} = \mathbb E_{\tau\tau}$ are diagonal matrices and $[\mathbb E_{\tau\sigma}, \mathbb E_{\sigma \tau}]_+ = \mathbb E_{\tau\tau} + \mathbb E_{\sigma\sigma}$. Hence, the equation \eqref{quan2} yields
\begin{gather*}
\operatorname{Tr}\big(\q^k\p^2\big)\Psi(Q)=
\hbar^2 \sum_{\sigma}z_\sigma^k \partial_{z_{\sigma}}^2 \Psi(Z)
 \\
\qquad{} +
\hbar^2 \sum_{\substack{\sigma ,\tau \\ \sigma\neq \tau}}z_\sigma^k \left(-\frac{\gamma([\mathbb E_{\tau \sigma},\mathbb E_{\sigma\tau}]_+)}{2(z_\sigma-z_\tau)^2}\Psi\big(\tilde{Z}\big) +
 \left(\frac{\partial_{z_\sigma}-\partial_{z_\tau}}{z_\sigma-z_\tau}\right)\gamma\big({\rm e}^{-A^{(1)}-A^{(2)}}\big)\Psi\big(\tilde{Z}\big)\right)\bigg|_{M=0}.
\end{gather*}
To complete the computation, we recall that under the assumption for the representation space~$\mathbb V_\kappa$ (see~\eqref{Vk}) we easily see that $\gamma( \mathbb E_{\sigma\tau} \mathbb E_{\tau\sigma}) = \kappa(\kappa+1) {\rm Id}_{\mathbb V_\kappa(0)}$. Recall also (Lemma~\ref{lemmazero}) that $\Psi$ evaluated on diagonal matrices takes values in the zero weight space $\mathbb V_\kappa(0)$.
Therefore we conclude that $\gamma([\mathbb E_{\tau\sigma},\mathbb E_{\sigma\tau}]_+)$ reduces simply to the multiplication by $2\kappa(\kappa+1)$.

This leads finally to the result
\begin{gather*}
\operatorname{Tr}\big(\q^k\p^2\big)\Psi(Q)= \bigg(\hbar^2 \!\sum_{\sigma}z^k_\sigma \partial_{z_{\sigma}}^2-\hbar^2 \kappa (\kappa+1)\sum_{\substack{\sigma ,\tau \\ \sigma\neq \tau}}\frac{z^k_\sigma}{(z_\sigma-z_\tau)^2}
+\hbar^2 \!\sum_{\substack{\sigma ,\tau \\ \sigma\neq \tau}}
 \frac{z_\sigma^k (\partial_{z_\sigma}- \partial_{z_\tau})}{z_\sigma-z_\tau}\bigg)\Psi(Z).
\end{gather*}
To make the last term more symmetric we can add and subtract $z_{\tau}^k\partial_{z_\tau}$ and obtain
\begin{gather*}
\operatorname{Tr}\big(\q^k\p^2\big)\Psi(Q)=
\bigg(\hbar^2 \!\sum_{\sigma}z^k_\sigma \partial_{z_{\sigma}}^2-
\hbar^2 \kappa (\kappa+1)
 \sum_{\substack{\sigma ,\tau \\ \sigma\neq \tau}} \frac{z^k_\sigma}{(z_\sigma-z_\tau)^2}
+\hbar^2\! \sum_{\substack{\sigma ,\tau \\ \sigma\neq \tau}}
 \frac{z_\sigma^k \partial_{z_\sigma}-z_\tau^k \partial_{z_\tau}} {z_\sigma-z_\tau}\bigg)\Psi(Z)
\\
\hphantom{\operatorname{Tr}\big(\q^k\p^2\big)\Psi(Q)=}{} -\hbar^2 \sum_{j=0}^{k-1} \sum_{\sigma} z_\sigma^j \sum_{\tau} z_\tau^{k-j-1} \partial_{z_\tau}\Psi(Z) +\hbar^2 k \sum_{\tau} z_\tau^{k-1} \pa_{z_\tau}\Psi(Z).
\end{gather*}
\end{Example}

Throughout the following section, we use this method to obtain the quantized Calogero--Painlev\'e {II}--{VI} Hamiltonian operators.

\subsubsection{Quantization of Calogero--Painlev\'e Hamiltonians I--VI}

We apply the quantum Hamiltonian KKS reduction explained in the previous sections to the list of Hamiltonians~\eqref{revisedCP}.
Moreover note that unlike the classic case, in the quantum case we have non-commutative operator-valued matrices $\p$ and $\q$ so that $\Tr(\p\q)$ is not equal to $\Tr(\q\p)$. For example, we have the following
\begin{gather*}
\operatorname{Tr}(\p\q)= \operatorname{Tr}(\q\p)+\hbar N^2, \nn \\
\operatorname{Tr}(\p\q\p)= \operatorname{Tr}\big(\q\p^2\big)+\hbar N \operatorname{Tr}(\p), \nn\\
\operatorname{Tr}(\q\p\q)= \operatorname{Tr}\big(\q^2 \p\big)+ \hbar N \operatorname{Tr}(\q),\nn \\
\operatorname{Tr}\big(\p\q^2\big)= \operatorname{Tr}\big(\q^2 \p\big)+2\hbar N \operatorname{Tr}(\q),\nn \\
\operatorname{Tr}\big(\p^2 \q^2\big) = \operatorname{Tr}\big(\q^2 \p^2\big)+2\hbar N \operatorname{Tr}(\q\p)+2\hbar \operatorname{Tr}(\q)\operatorname{Tr}(\p)+\hbar^2 N\big(1+N^2\big) .
\end{gather*}

We start from the Hamiltonian operator corresponding to Calogero--Painlev\'e~II (more complicated case than Calogero--Painlev\'e~I) and we apply it to the $\gamma$-equivariant wave func\-tion~$\Psi(Q)$, and then we apply the quantization \eqref{quan} to the result:
\begin{gather}
\tilde{H}_{\rm II}\Psi(Q)=
\operatorname{Tr}\le(\frac{\p^2}{2}-\frac{1}{2}\le(\q^2+\frac{t}{2}\ri)^2-\theta \q\ri)\Psi(Q)
\nonumber\\
\hphantom{\tilde{H}_{\rm II}\Psi(Q)}{} =
\le(\frac{\hbar^2}{2}\sum_{\rho,\sigma}\partial_{q_{\sigma \rho}}\partial_{q_{\rho \sigma}}-\frac{1}{2}\sum_{\rho,\sigma}\le(\big(\delta_{\sigma\rho}^2q_{\rho\sigma}^2\big)+\frac{t}{2}\ri)^2-\theta \delta_{\sigma\rho}q_{\rho\sigma}
\ri)\Psi(Q).
\label{227}
\end{gather}
Following the same logic used in the illustrative example of the operator $\Tr(\p\q\p)$ we can continue the computation
\begin{gather}
\eqref{227}=
\frac{\hbar^2}{2}\!\sum_{\rho,\sigma}\partial_{q_{\sigma \rho}}\partial_{q_{\rho \sigma}} \gamma\big({\rm e}^{-A^{(1)}-A^{(2)}}\big)\Psi\big(\wt{Z}\big)-\frac{1}{2}\sum_{\rho}\!\left(z_{\rho}^2+\frac{t}{2}\right)^2\!\Psi(Z)-\theta \sum_{\rho}z_{\rho}\Psi(Z),\!\!\!
\label{228}
\end{gather}
where $\wt Z = Z+\operatorname{diag}\Big(\sum_{d}\frac{M_{\bullet d} M_{d\bullet}}{z_\bullet- z_d}\Big)_{\bullet=1}^N
 $ as in~\eqref{diag2}.
Note that the terms involving only the multiplication operator $\q$ can be directly evaluated at $Q= Z$ setting $M=0$. The second-order operator $\sum_{\rho,\sigma} \pa_{q_{\sigma\rho}} \pa_{q_{\rho\sigma}}$ in terms of~$Z$,~$M$ becomes the operator $\sum_{\rho} \pa_{z_\rho}^2 + \sum_{\substack{\rho,\sigma \\ \rho\neq \sigma}}\partial_{M_{\sigma \rho}}\partial_{M_{\rho \sigma}}$; once again, the part involving the derivatives with respect to $z_\rho$ can be directly evaluated at $Q=Z$ while we postpone the evaluation of the part involving the derivatives in $M_{\rho,\sigma}$:
\begin{gather*}
\eqref{228}=
\frac{\hbar^2}{2}\sum_{\rho}\partial_{z_{\rho}}^2\Psi(Z)+\frac{\hbar^2}{2}
\underbrace{
\sum_{\substack{\rho,\sigma \\ \rho\neq \sigma}}\partial_{M_{\sigma \rho}}\partial_{M_{\rho \sigma}}\big(\gamma\big({\rm e}^{-A^{(1)}-A^{(2)}}\big)\Psi\big(\wt{Z}\big)\big)
}_{*} \\
\hphantom{\eqref{228}=}{}
-\frac{1}{2}\sum_{\rho}\left(z_{\rho}^2+\frac{t}{2}\right)^2\Psi(Z) -\theta \sum_{\rho}z_{\rho}\Psi(Z).
\end{gather*}

The term indicated by the asterisk is dealt with in complete analogy to the similarly marked term in~\eqref{quan2}. We thus obtain
\begin{gather*}
\frac{\hbar^2}{2}\sum_{\rho}\partial_{z_{\rho}}^2\Psi(Z)+\frac{\hbar^2}{2}\sum_{\substack{\rho,\sigma \\
 \rho\neq \sigma}}\left(-\frac{1}{2}\gamma\le(\frac{[\mathbb E_{\rho \sigma},\mathbb E_{\sigma \rho}]_+}{(z_\sigma-z_\rho)^2}\ri)
 +
 \gamma\big({\rm e}^{-A^{(1)}-A^{(2)}}\big)\le(\frac{\partial_{z_\rho}-\partial_{z_\sigma}}{z_\rho-z_\sigma}\ri)\right)\Psi\big(\wt{Z}\big) \\
\qquad{} -\frac{1}{2}\sum_{\rho}\left(z_{\rho}^2+\frac{t}{2}\right)^2\Psi(Z)-\theta \sum_{\rho}z_{\rho}\Psi(Z).
\end{gather*}
 Hence, putting these all together, we obtain
\begin{gather*} 
\tilde{H}_{\rm II}\Psi(Q)
=
\frac{\hbar^2}{2}\sum_{\rho}\partial_{z_{\rho}}^2\Psi(Z)-\frac{\hbar^2 \kappa(\kappa+1)}{2}\sum_{\substack{\rho,\sigma \\ \rho\neq \sigma}} \frac{1}{(z_\rho-z_\sigma)^2}\Psi(Z)
+ \frac{\hbar^2}{2}\sum_{\substack{\rho,\sigma \\ \rho\neq \sigma}}\frac{\partial_{z_\rho}-\partial_{z_\sigma}}{z_\rho-z_\sigma}\Psi(Z) \\
\hphantom{\tilde{H}_{\rm II}\Psi(Q)=}{}
-\frac{1}{2}\sum_{\rho}\left(z_{\rho}^2+\frac{t}{2}\right)^2\Psi(Z)-\theta\sum_{\rho}z_{\rho}\Psi(Z).
\end{gather*}
The equation takes a more convenient form if we apply to the wave function a gauge transformation of the form
\begin{equation*}
\Psi(Z)=\exp\le[-\frac{1}{\hbar}\sum_{\alpha}\le(\frac{z_{\alpha}^3}{3}+\frac{t}{2}z_{\alpha}\ri)\ri]\Phi(Z).
\end{equation*}
As a result, the Schr\"odinger equation $\hbar \partial_t \Psi(Z)=\wt H_{\rm II}\Psi(Z)$ is transformed into the one with the new Hamiltonian
\begin{gather}
\hat{H}_{\rm II}=\frac{\hbar^2}{2}\sum_{\substack{\rho,\sigma \\ \rho\neq \sigma}}\frac{\partial_{z_\sigma}-\partial_{z_\rho}}{z_\sigma-z_\rho}-\frac{\hbar^2 \kappa(\kappa+1)}{2}\sum_{\substack{\rho,\sigma \\ \rho\neq \sigma}}\frac{1}{\le(z_\sigma-z_\rho\ri)^2}
+\frac{\hbar^2}{2}\sum_\rho \partial_{z_\rho}^2-\hbar\sum_\rho\le(z_\rho^2+\frac{t}{2}\ri)\partial_{z_\rho} \nonumber\\
\hphantom{\hat{H}_{\rm II}=}{}
+\le(\frac{1}{2}-\theta-\hbar N\ri)\sum_\rho z_\rho.\label{H2_2}
\end{gather}

As a result of a similar computation for Calogero--Painlev\'e I Hamiltonian operator in~\eqref{IsoHamil}, and each of the Hamiltonian operators in the system \eqref{revisedCP} (no gauge transformation required), we obtain the following Hamiltonian system for the quantized Calogero--Painlev\'e system:
\begin{gather} \label{H1}
\tilde{H}_{\rm I} = \frac{\hbar^2}{2}\sum_{\substack{\rho,\sigma \\ \rho\neq \sigma}}\frac{\partial_{z_\sigma}-\partial_{z_\rho}}{z_\sigma-z_\rho}-\frac{\hbar^2 \kappa(\kappa+1)}{2}\sum_{\substack{\rho,\sigma \\ \rho\neq \sigma}}\frac{1}{\le(z_\sigma-z_\rho\ri)^2}
+\frac{\hbar^2}{2}\sum_\rho \partial_{z_\rho}^2 -\sum_{\rho} \le(\frac{z_{\rho}^3}{2}+\frac{tz_{\rho}}{4}\ri),\!\!\!
\\
t\tilde{H}_{\rm III} =
\hbar^2\sum_{\substack{\rho,\sigma \\ \rho\neq \sigma}}\frac{z_\sigma^2\partial_{z_\sigma}-z_\rho^2\partial_{z_\rho}}{z_\sigma-z_\rho}+\sum_\rho \big(\hbar^2 z_\rho^2\partial_{z_\rho}^2-\hbar\big(z_\rho^2-(2\hbar-\theta_0+\theta_1)z_\rho-t\big)\partial_{z_\rho} \nn \\
\hphantom{t\tilde{H}_{\rm III} =}{}-(\hbar N+\theta_1) z_\rho\big)
-\frac{\hbar^2\kappa(\kappa+1)}{2} \sum_{\substack{\rho,\sigma \\ \rho\neq \sigma}}\frac{z_\rho^2+z_\sigma^2}{(z_\sigma-z_\rho)^2}+\frac{\hbar^2N\big(1+N^2\big)}{2},\label{H3}
\\
\tilde{H}_{\rm IV}
=
\hbar^2\sum_{\substack{\rho,\sigma \\ \rho\neq \sigma}}\frac{z_\sigma\partial_{z_\sigma}-z_\rho\partial_{z_\rho}}{z_\sigma-z_\rho}+\sum_\rho \big(\hbar^2z_\rho\partial_{z_\rho}^2-\hbar\big(z_\rho^2+tz_\rho-\theta_0-\hbar\big)\partial_{z_\rho}\big) \nn \\
\hphantom{\tilde{H}_{\rm IV}}{} -\frac{\hbar^2\kappa(\kappa+1)}{2} \sum_{\substack{\rho,\sigma \\ \rho\neq \sigma}}\frac{z_\rho+z_\sigma}{(z_\sigma-z_\rho)^2} -(\hbar N+\theta_0+\theta_1)\sum_\rho z_\rho-t\hbar N^2, \label{H4}
\\
t\tilde{H}_{\rm V}=
\hbar^2\sum_{\substack{\rho,\sigma \\ \rho\neq \sigma}}\frac{z_\sigma(z_\sigma-1)\partial_{z_\sigma}-z_\rho(z_\rho-1)\partial_{z_\rho}}{z_\sigma-z_\rho}+\hbar^2\sum_\rho z_\rho(z_\rho-1)\partial_{z_\rho}^2 \nn \\
\hphantom{t\tilde{H}_{\rm V}=}{}
-\frac{\hbar^2 \kappa(\kappa+1)}{2} \sum_{\substack{\rho,\sigma \\ \rho\neq \sigma}}\frac{z_\rho(z_\rho-1)+z_\sigma(z_\sigma-1)}{(z_\sigma-z_\rho)^2}\nonumber\\
\hphantom{t\tilde{H}_{\rm V}=}{}
+\hbar \sum_\rho\big(t z_\rho^2+\big(2\hbar+\theta_0-\theta_2-t\big) z_\rho+\theta_2-\hbar \big)\partial_{z_\rho}
+ t(\hbar N+\theta_0+\theta_1)\sum_\rho z_\rho\nonumber\\
\hphantom{t\tilde{H}_{\rm V}=}{}
+\le((\theta_0-\theta_2-t)N^2\hbar+\frac{N\hbar^2\big(1+N^2\big)}{2}\ri),\label{H5}
\\
t(t-1)\tilde{H}_{\rm VI}=
\hbar^2\sum_{\substack{\rho,\sigma \\ \rho\neq \sigma}}\frac{z_\sigma(z_\sigma-1)(z_\sigma-t)\partial_{z_\sigma}-z_\rho(z_\rho-1)(z_\rho-t)\partial_{z_\rho}}{z_\sigma-z_\rho}\nonumber\\
\hphantom{t(t-1)\tilde{H}_{\rm VI}=}{}
+\hbar^2\sum_\rho z_\rho(z_\rho-1)(z_\rho-t)\partial_{z_\rho}^2 \nn
\\
\hphantom{t(t-1)\tilde{H}_{\rm VI}=}{}
+\hbar\sum_\rho \big((3\hbar-\theta)z_\rho^2
+ \big(-\hbar(1+t)+t(\theta_0+\theta_1)+\theta_0+\theta_t\big)z_\rho+t(\hbar -\theta_0)\big)\partial_{z_\rho} \nn\\
\hphantom{t(t-1)\tilde{H}_{\rm VI}=}{} -\frac{\hbar^2 \kappa(\kappa+1)}{2} \sum_{\substack{\rho,\sigma \\ \rho\neq \sigma}}\frac{z_\rho(z_\rho-1)(z_\rho-t)+z_\sigma(z_\sigma-1)(z_\sigma-t)}{(z_\sigma-z_\rho)^2} \nn
\\
\hphantom{t(t-1)\tilde{H}_{\rm VI}=}{}
+\left(N^2\hbar^2-\theta N\hbar-\frac{1}{4}\big(k^2-\theta^2\big)+(N-1)\kappa(\kappa+1)\hbar^2\right)\sum_\rho z_\rho \nn \\
\hphantom{t(t-1)\tilde{H}_{\rm VI}=}{}
 -\frac{N^3\hbar^2}{2}+t\hbar N^2(\theta_0+\theta_1)+\hbar N^2(\theta_0+\theta_t)-\frac{\hbar^2 N(N-1)\kappa(\kappa+1)}{2}.\label{H6}
\end{gather}

\setcounter{footnote}{1}
\section{Generalization of the quantum Painlev\'e equations}\label{sec3}
Nagoya in \cite{Nag}, introduced the Hamiltonian operators corresponding to the quantum Painlev\'e equations in the case of a single particle, satisfying the Schr\"odinger equation
\begin{equation}
\hbar \frac{\partial}{\partial t}\Phi (z,t)=H_J \left(z, \hbar \frac{\partial}{\partial z},t\right)\Phi(z,t), \qquad J={\rm II},{\rm III},{\rm IV},{\rm V},{\rm VI},\label{Schreq}
\end{equation}
where Hamiltonian operators $H_J$ are obtained from the polynomial Hamiltonian operators of the Painlev\'e equations by substituting the operators~$z$, $\hbar \frac{\partial}{\partial z}$ into the canonical coordinates. These operators are defined as\footnotetext{The coefficient of $z$ in \cite{Nag} is shifted by the parameter $b$, however, based on our generalization, the generated coefficient of $z$ is only $a$.}\addtocounter{footnote}{-1}
\begin{gather}
H_{\rm II} =
\frac{1}{2}(\hbar \partial_{z})^2-\left(z^2+\frac{t}{2}\right)\hbar\partial_{z}+a z, \nn \\
t H_{\rm III} = z^2(\hbar \partial_z)^2- (z^2+bz+t)\hbar \partial_z +a z, \footnotemark \nn\\
H_{\rm IV} =
z(\hbar \partial_{z})^2-\big(z^2+tz+b\big)\hbar\partial_{z}+a (z+t), \nn\\
tH_{\rm V}
=
z(z-1)(\hbar\partial_{z})^2+\big(tz^2-(b+c+t)z+b\big)\hbar\partial_{z}
+a(b+c-a+\hbar+t-tz), \nn\\
t(t-1)H_{\rm VI}=
z(z-1)(z-t)(\hbar\partial_{z})^2-\big((a+b)(z-1)(z-t)
+cz(z-t) +dz(z-1)\big)\hbar\partial_{z}\nonumber\\
\hphantom{t(t-1)H_{\rm VI}=}{} +(b+c+d+\hbar)a(z-t).\label{NagHamil}
\end{gather}
He showed that the quantum Hamiltonians \eqref{NagHamil} admit special solutions in integral form when the parameters take certain specific values: in fact, the wave function $\Phi$ in~\eqref{Schreq} can be taken in the following form
\begin{equation}
\Phi_m^J(z,t)=\int_{\Gamma} \prod_{1\leq i<j\leq m}(u_i-u_j)^{2\hbar} \prod_{i}(z-u_i)\Theta_J(u_i,t)\d u_i.
\label{phi1p}
\end{equation}
Here $\Gamma = \prod_j \gamma_j$ is a cartesian product of ``admissible'' contours $u_j\in \gamma_j$ that can be chosen on the Riemann surface of the ``master'' function $\Theta_{J}$. By this we mean that
\begin{itemize}\itemsep=0pt
\item $\Theta_J(u_j)$ is single valued along $\gamma_j$;
\item $\int_{\gamma_j}u_j^k \Theta_J(u_j)\d u_j$ is a convergent integral and not identically zero (as an expression in $k\in \mathbb N$);
\item The contours are pairwise non-intersecting if $\hbar \not\in \frac 1 2\mathbb N$.
\end{itemize}
For example for $J={\rm II}$ the contours can be taken as contours starting from infinity along one of the three directions $\arg(u_j)=\frac {2\pi}3 k$, $k=0,1,2$ and ending at infinity along any of the remaining ones. We could also take a circle, but then the Cauchy theorem would imply that the integral $\int u^\ell \Theta_{\rm II}(u)\d u$ is zero.
The requirement that the different $\gamma_j$'s do not intersect is due to the fact that if $\hbar\not\in \frac 1 2 \mathbb N$ then the power of the Vandermonde term in the integrand~\eqref{phi1p} yields a~non-single valued function.

The master functions $\Theta_J(u_i,t)$ ($J={\rm II},{\rm III},{\rm IV},{\rm V},{\rm VI}$) are the weight functions defined below:
\begin{gather}
\Theta_{\rm II} = \exp \left(-\left( u_it+\frac{2}{3}u_i^3\right)\right), \nn \\
\Theta_{\rm III} =
u_i^{-b-1} \exp \left(\frac{t}{u_i}-u_i\right), \nn \\
\Theta_{\rm IV} =
u_i^{-b-1} \exp \left(-\left(u_it+\frac{u_i^2}{2}\right)\right), \nn
\\
\Theta_{\rm V} =
u_i^{-b-1}(1-u_i)^{-c-1} \exp (u_it), \nn \\
\Theta_{\rm VI} =
u_i^{-a-b-1}(1-u_i)^{-c-1}(t-u_i)^{-d}.\label{masters}
\end{gather}

With the positions \eqref{masters} and formula \eqref{phi1p} the functions $\Phi_m^J$ satisfy the Schr\"odinger equations~\eqref{Schreq} provided that the parameters $a$, $b$, $c$, $d$ satisfy the following condition
\begin{equation*}
\begin{cases}
a=m\hbar \quad \text{and} \quad b+c+d=(m-1)\hbar, & J={\rm VI},\\
a=m\hbar, & J={\rm II},{\rm III},{\rm IV},{\rm V}.
\end{cases}
\end{equation*}
We want to generalize this result of Nagoya's to our multi-particle quantum Hamiltonians \eqref{H2_2}, \eqref{H3}--\eqref{H6} by providing a multi-particle extension of the integral formul\ae~\eqref{phi1p}.

\subsection[Integral representations for the quantum Painlev\'e--Calogero Schr\"odinger wave functions]{Integral representations for the quantum Painlev\'e--Calogero Schr\"odinger\\ wave functions}

We now explain the general approach behind the extension of Nagoya's formul\ae. We start from an Ansatz of the form
\begin{equation}
\Phi(\vec z;t)=\int_{\Gamma} \prod_{1\leq i<j\leq m}(u_i-u_j)^{2\hbar} \prod_{\rho=1}^N\prod_{i=1}^m(z_{\rho}-u_i)\Theta_J(u_i) \d u_i.
\label{phiNp}
\end{equation}
Observe that, similarly to Nagoya's result, these are polynomials in the $z_\rho$'s of total degree~$Nm$ and of degree $m$ in each of the variables $z_\rho$.
With the Ansatz~\eqref{phiNp} in place, we verify by a direct calculation on a case-by-case basis, that they satisfy a multi-variate generalization of~\eqref{Schreq} and identify the corresponding Hamiltonian operator.
The result of these computations, whose details are reported in the appendix section, are summarized in the following theorem (for readability, the range of greek indices is assumed to be $1,\dots, N$ without explicit mention).
\begin{Theorem}
The equation \eqref{phiNp}, with $\Theta_J$ given by \eqref{masters}, is a solution to the multi-variable version of Schr\"odinger equation \eqref{Schreq}, where the Hamiltonians $H_J$, $J={\rm II}, {\rm III}, {\rm IV}, {\rm V}, {\rm VI}$ are given by the following operators:
\begin{gather}
H_{\rm II}= \frac{\hbar}{2}\underset{\rho\neq\sigma}{\sum_{\rho,\sigma}} \frac{\partial_{z_{\rho}}-\partial_{z_{\sigma}}}{z_{\rho}-z_{\sigma}} +\frac{\hbar^2}{2} \sum_{\rho}\partial_{z_{\rho}}^2 -\hbar\sum_{\rho} \left(z_{\rho}^2+\frac{t}{2}\right)\partial_{z_{\rho}} +m\hbar \sum_{\rho} z_{\rho},
\label{CP2}
\\
tH_{\rm III}= \hbar\underset{\rho\neq\sigma}{\sum_{\rho,\sigma}}\frac{z_{\rho}^2\partial_{z_{\rho}}-z_{\sigma}^2\partial_{z_{\sigma}}}{z_{\rho}-z_{\sigma}}+ \sum_{\rho}\big(\hbar^2 z_{\rho}^2\partial_{z_{\rho}}^2-\hbar \big(z_{\rho}^2+(b+N-1)z_{\rho}+t\big)\partial_{z_{\rho}}+m\hbar z_{\rho}\big),
\label{CP3}
\\
H_{\rm IV}= \hbar\underset{\rho\neq\sigma}{\sum_{\rho,\sigma}}\frac{z_{\rho}\partial_{z_{\rho}}-z_{\sigma}\partial_{z_{\sigma}}}{z_{\rho}-z_{\sigma}}+ \sum_{\rho}\big(\hbar^2 z_{\rho}\partial_{z_{\rho}}^2-\hbar \big(z_{\rho}^2+tz_{\rho}+b\big)\partial_{z_{\rho}}+m\hbar z_{\rho}\big) + \hbar Nmt,
\label{CP4}
\\
tH_{\rm V}=
 \hbar\underset{\rho\neq\sigma}{\sum_{\rho,\sigma}}\frac{z_{\rho}(z_{\rho}-1)\partial_{z_{\rho}}-z_{\sigma}(z_{\sigma}-1)\partial_{z_{\sigma}}}{z_{\rho}-z_{\sigma}}
+
\hbar Nm\big(b+c+t -\hbar(m-1)- N+1\big)
 \nn\\
\hphantom{tH_{\rm V}=}{}
+\sum_{\rho}\big(\hbar^2 z_{\rho}(z_{\rho}-1)\partial_{z_{\rho}}^2+\hbar \big(tz_{\rho}^2-(b+c+t)z_{\rho}+b\big)\partial_{z_{\rho}}
- m\hbar tz_{\rho}\big),
\label{CP5}
\\
t(t-1)H_{\rm VI}=
 \hbar\underset{\rho\neq\sigma}{\sum_{\rho,\sigma}}\frac{z_{\rho}(z_{\rho}-1)(z_{\rho}-t) \partial_{z_{\rho}}-z_{\sigma}(z_{\sigma}-1)(z_{\sigma}-t)\partial_{z_{\sigma}}}{z_{\rho}-z_{\sigma}}\nonumber\\
 \hphantom{t(t-1)H_{\rm VI}=}{}
+ \sum_{\rho}\hbar^2 z_{\rho}(z_{\rho}-1)(z_{\rho}-t)\partial_{z_{\rho}}^2
\nn\\
\hphantom{t(t-1)H_{\rm VI}=}{} -\hbar \sum_{\rho}\left((a+b)(z_{\rho}-1)(z_{\rho}-t)+cz_{\rho}(z_{\rho}-t)+(d+N-1)z_{\rho}(z_{\rho}-1)\right)\partial_{z_{\rho}}
\nn\\
\hphantom{t(t-1)H_{\rm VI}=}{} - \hbar m(N-1-\hbar m)\sum_{\rho}z_{\rho}
- \hbar m N(\hbar m+1-N)t.
\label{CP6}
\end{gather}
\end{Theorem}

{\bf Quantum Calogero Hamiltonians and integral solutions.}
We now examine the relationship between the Hamiltonians (differential operators) \eqref{CP2}--\eqref{CP6}, whose integral solutions are given~\eqref{phiNp} and the quantum Calogero Hamiltonians \eqref{H2_2}, \eqref{H3}--\eqref{H6} constructed by canonical quantization of the non-commutative Hamiltonians of the classical isomonodromic noncommutative equations of~\cite{BerCafRub}.
By way of this identification, we identify the parameter $\kappa$ of the quantum radial reduction up to an appropriate gauge transformation of the Hamiltonians. At the end of this computation, we can see that the Hamiltonians
\eqref{H2_2}, \eqref{H3}--\eqref{H6} reduce to
\eqref{CP2}--\eqref{CP6} for particular choice of the parameters. The results are summarized in Table~\ref{tablecalogero}.
Consider for example \eqref{H2_2} and \eqref{CP2}:
\begin{gather*}
\tilde{H}_{\rm II}=\frac{\hbar^2}{2}\sum_{\substack{\rho,\sigma \\ \rho\neq \sigma}}\frac{\partial_{z_\sigma}-\partial_{z_\rho}}{z_\sigma-z_\rho}-\frac{\hbar^2 \kappa(\kappa+1)}{2}\sum_{\substack{\rho,\sigma
\\ \rho\neq \sigma}}\frac{1}{\le(z_\sigma-z_\rho\ri)^2}
+\frac{\hbar^2}{2}\sum_\rho \partial_{z_\rho}^2-\hbar\sum_\rho\le(z_\rho^2+\frac{t}{2}\ri)\partial_{z_\rho}\nn\\
\hphantom{\tilde{H}_{\rm II}=}{} +\le(\frac{1}{2}-\theta-\hbar N\ri)\sum_\rho z_\rho,
\\
H_{\rm II}= \frac{\hbar}{2}\underset{\rho\neq\sigma}{\sum_{\rho,\sigma}} \frac{\partial_{z_{\rho}}-\partial_{z_{\sigma}}}{z_{\rho}-z_{\sigma}} +\frac{\hbar^2}{2} \sum_{\rho}\partial_{z_{\rho}}^2 -\hbar\sum_{\rho} \le(z_{\rho}^2+\frac{t}{2}\ri)\partial_{z_{\rho}} +m\hbar \sum_{\rho} z_{\rho}.
\end{gather*}

The first observation is that in the two Hamiltonians the second and first-order differential parts have different powers of $\hbar$; they coincide only for $\hbar =1$ (we exclude the trivial case $\hbar=0$). We will comment on the $\hbar=1$ case later on.

The second observation is about the Calogero-like potential term in~\eqref{H2_2} which is absent in~\eqref{CP2}; the term disappears for $\kappa=0,-1$. These values mean that the ${\rm GL}_n$ representation in~$\mathbb V$ is the trivial one in the radial quantization scheme.

Then the other parameters in the equations can be easily matched and then we need $\theta=\frac{1}{2}-N-m$.

 A more general family of matching between the two sets of operators is obtained by gauging the Hamiltonians \eqref{CP2}--\eqref{CP6} by means of an appropriate power of the Weyl denominator, namely, the Vandermonde determinant. We indicate this in the following lemma:
\begin{Lemma}\label{le31}
Let $\Delta(\vec z)=\prod_{1\leq \alpha<\beta\leq N}(z_\alpha -z_\beta)$ be the Vandermonde polynomial in $\Vec{z}$. Then the action of the quantum Hamiltonian operator $H_{\rm II}$ on the generalized wave functions, is equivalent to the action of $\Delta^{-R}\wt H_{\rm II} \Delta^R$:
\begin{gather} \label{HHtilde}
\Delta^{-R}\wt H_{\rm II}\Delta^R \Phi(Z)=H_{\rm II} \Phi(Z)
\end{gather}
for $R$ and the scalar $\kappa$ determined as one of the two choices below:
\begin{gather*} 
\left(R=\frac{1}{\hbar}-1 ,\, \kappa=\frac{1}{\hbar}-1\right), \qquad
\left(R=\frac 1\hbar -1 , \,
 \kappa = -\frac 1\hbar \right) ,\qquad \hbar \neq 0.
 \end{gather*}
Under this circumstances, the parameter $\theta$ is determined as $($in either case$)$:
\begin{gather*}
\theta = \hbar(1-m) + N(1-2\hbar)-\frac{1}{2}.
\end{gather*}
\end{Lemma}
\begin{proof} The proof is obtained by a direct computation of the equations in~\eqref{HHtilde}.
\end{proof}

Lemma~\ref{le31} is actually the manifestation of a more general result which is contained in the following theorem:
\begin{Theorem}\label{thmgauge}
Define the two sequences of differential operators
\begin{gather*}
H_a := \hbar^2 \sum_{\rho} z_\rho^a \pa_{z_\rho}^2 + 2\hbar \sum_{\substack{\rho,\sigma\\ \rho<\sigma}} \frac {z_\rho^a \pa_{z_\rho} - z_\sigma^a\pa_{z_\sigma}}{z_\rho-z_\sigma}, \\
\wt H_a := \hbar^2 \sum_{\rho} z_\rho^a \pa_{z_\rho}^2 + 2\hbar^2 \sum_{\substack{\rho,\sigma \\ \rho<\sigma}} \frac {z_\rho^a \pa_{z_\rho} - z_\sigma^a\pa_{z_\sigma}}{z_\rho-z_\sigma} - \hbar^2\kappa(\kappa+1) \sum_{\substack{\rho,\sigma \\ \rho<\sigma}} \frac {z_\rho^a + z_\sigma^a}{(z_\rho-z_\sigma)^2} \\
\hphantom{\wt H_a :=}{} +
\begin{cases}
0, & a= 0,1,\\
\dfrac {N(N-1)(N-2)}3, & a=2,\\
\ds (N-1)(N-2)\sum_{\rho} z_\rho, & a=3 .
\end{cases}
\end{gather*}
Then we have the following identity
\[
H_a = \Delta^{-R}\circ \wt H_a \circ \Delta^R
\]
provided that
$R= \frac 1 \hbar-1$ and
$\kappa = \frac 1 \hbar -1$ or $\kappa=-\frac 1 \hbar$.
\end{Theorem}

We do not report here the straightforward proof.
Using Theorem~\ref{thmgauge} in the various cases of the Hamiltonian Calogero--Painlev\'e operators and direct computations, allows to express the Hamiltonians \eqref{H2_2}, \eqref{H3}--\eqref{H6} as special cases of \eqref{CP2}--\eqref{CP6} and allows us to identify the parameters $\theta_0$, $\theta_1$, $\theta_2$, $\theta_t$, $\theta$, $k$, $a$, $b$, $c$ and $d$. The result of the computations is summarized in Table \ref{tablecalogero}.

\begin{table}[h!]\footnotesize
\caption{The correspondence of parameters between the differential operators \eqref{CP2}--\eqref{CP6} and \eqref{H2_2}, \eqref{H3}--\eqref{H6} with the relationship for the parameters.}\label{tablecalogero}

\vspace{1mm}

\centering
\renewcommand{\arraystretch}{1.7}
\begin{tabular}{||c||c|c|}
\hline
&$H\Psi(Z)=\wt H \Psi(Z)$ & $\Delta^{-R} \wt H \Delta^R \Psi(Z)= H \Psi(Z)$ \\
\hline \hline
Calogero--Painlev\'e II & $\begin{matrix} \kappa =0, \quad \hbar=1, \\ \theta =\frac{1}{2}-N-m \end{matrix}$ & $\begin{matrix} \kappa =\frac{1}{\hbar}-1 \quad \hbox{or} \quad \kappa =-\frac{1}{\hbar}, \quad R=\frac{1}{\hbar}-1, \\ \theta = \hbar(1-m) + N(1-2\hbar)-\frac{1}{2} \end{matrix}$ \\
\hline
Calogero--Painlev\'e III & $\begin{matrix} \kappa =0, \quad \hbar=1, \\ \theta_0=b-m+1, \\ \theta_1=-N-m \end{matrix}$ & $\begin{matrix} \kappa =\frac{1}{\hbar}-1 \quad \hbox{or} \quad \kappa =-\frac{1}{\hbar}, \quad R=\frac{1}{\hbar}-1, \\ \theta_0=b+\hbar(1-m), \\ \theta_1=-\hbar(m+1)-N+1 \end{matrix}$ \\
\hline
Calogero--Painlev\'e IV & $\begin{matrix} \kappa =0, \quad \hbar=1, \\ \theta_0=-b-1, \\ \theta_1=b+1-m-N \end{matrix}$ & $\begin{matrix} \kappa =\frac{1}{\hbar}-1 \quad \hbox{or} \quad \kappa =-\frac{1}{\hbar}, \quad R=\frac{1}{\hbar}-1, \\ \theta_0=-b-\hbar, \\ \theta_1=b+1-N- m\hbar \end{matrix}$\\
\hline
Calogero--Painlev\'e V & $\begin{matrix} \kappa =0, \quad \hbar=1, \\ \theta_0=-c-1, \\ \theta_1=-N-m+c+1, \\ \theta_2= b+1 \end{matrix}$ & $\begin{matrix} \kappa =\frac{1}{\hbar}-1 \quad \hbox{or} \quad \kappa =-\frac{1}{\hbar}, \quad R=\frac{1}{\hbar}-1, \\ \theta_0=-c-\hbar, \\ \theta_1=c+1-N-m \hbar, \\ \theta_2=b+\hbar \end{matrix}$\\
\hline
Calogero--Painlev\'e VI & $\begin{matrix} \kappa =0, \quad \hbar=1, \\ \theta_0=a+b+1, \\
\theta_1=c+1, \\ \theta_t=d+N, \\
k=\pm \sqrt{A} \\
A=(\theta-2N)^2+4 m(N-1-m) \end{matrix}$ & $\begin{matrix} \kappa =\frac{1}{\hbar}-1 \quad \hbox{or} \quad \kappa =-\frac{1}{\hbar}, \quad R=\frac{1}{\hbar}-1, \\ \theta_0=a+b+\hbar, \\ \theta_1=c+\hbar, \\ \theta_t=d+N+\hbar-1, \\ k=\pm \sqrt{B} \\
 B=(\theta-2N\hbar)^2+4\hbar m(N-1-\hbar m)\\
 {}+4(N-1)(1-\hbar) +N(N-1)(1-\hbar)(3\hbar-\theta) \end{matrix}$\\ \hline
\end{tabular}

\end{table}

With these values of the parameters then quantum Calogero--Painlev\'e and generalized quantum Painlev\'e equations {II}--{VI} yield the same Hamiltonians, and the integral representation~\eqref{phiNp} provides solutions of the corresponding Schr\"odinger equations~\eqref{Schreq}.

We now briefly comment on the value of the Heisenberg constant $\hbar=1$: observing the original integral representation of Nagoya~\eqref{phi1p} and the generalized one~\eqref{phiNp}, we see that for $\hbar=1$ the integrand contains the square of the Vandermonde determinant of the variables~$u_j$. This type of expression is very familiar in the context of random matrices~\cite{Mehta}: it is the Jacobian of the change of variables from the Lebesgue measure on Hermitian matrices (or normal matrices) to the unitary-radial coordinates. Specifically, if $M$ is a Hermitean matrix of size $m\times m$, we write it as $M = V D V$ with $D =\operatorname{diag} (u_1,\dots, u_m)$ and $V\in U(m,\C)$, then the Lebesque measure (up to inessential multiplicative constant) is
\[
\d M = \Delta(u)^2 \d V \prod_j \d u_j ,
\]
where $\Delta(u) = \prod_{i<j} (u_i-u_j)$. This expression is also valid if the $u_j$'s are allowed to take complex values along specified curves (but now $\d M$ is the measure on $m\times m$ normal matrices).
This allows us to rewrite the integral formul\ae\ as matrix integrals; for example for Painlev\'e II we have
\begin{gather}
\Phi_m^{\rm II}(\vec z, t) = \int \prod_{\rho} \det(z_\rho - M) {\rm e}^{-\Tr\le(\frac 2 3 M^3 + tM\ri)} \d M,
\label{matrixint}
\end{gather}
which expresses the wave function as the expectation value of the product of characteristic polynomials. Similar expressions hold for the other cases.
Therefore it appears that the class of generalized Nagoya solutions~\eqref{phiNp} and the wave-functions of the quantum Calogero--Painlev\'e Hamiltonians intersect on the class of solutions that are related to matrix integrals of the form~\eqref{matrixint} only (and their similar expressions for the other master functions~\eqref{masters}).

We conclude by observing that we can write \eqref{matrixint} in terms of the matrix variable $Q$ (with eigenvalues $z_1,\dots, z_N$) as follows
\[
\Phi_m^{\rm II}(Q, t)=\int \prod_{\rho} \det(Q\otimes \1_m -\1_N\otimes M) {\rm e}^{-\Tr\le(\frac 2 3 M^3 + tM\ri)} \d M.
\]

For different values of $\hbar$ we need to use the substitution indicated in Theorem~\ref{thmgauge}; the resulting integral representation is what is known as a ``$\beta$-integral''. Thus, we obtain the following
\begin{Theorem}
Let $\Phi_J(\vec z;t)$ be the wave function \eqref{phiNp} solving the Schr\"odinger equations with Hamiltonians \eqref{CP2}--\eqref{CP6}. Then $\Psi_J(\vec z,t)= \Delta^{ \frac 1 \hbar -1} \Phi_J(\vec z;t)$ is a solution of the corresponding Painlev\'e $J$ Hamiltonians $(J={\rm II},\dots, {\rm VI})$ \eqref{H2_2}, \eqref{H3}--\eqref{H6}, i.e.,
\[
\Psi_J(\vec z;t)=
\Delta(\vec z)^{\frac 1 \hbar-1} \int \Delta(\vec u)^{2\hbar} \prod_{\rho,j} (z_\rho - u_j) \Theta_{J}(u_j,t)\d u_i
\]
give solutions to the multi-particle quantized Calogero--Painlev\'e equations with Hamiltonians \eqref{H2_2}, \eqref{H3}--\eqref{H6} for $\kappa=-\frac 1 \hbar, \frac 1 \hbar-1$ $($respectively$)$ and the values of parameters indicated in Table~{\rm \ref{tablecalogero}}.
\end{Theorem}

\section{Conclusions}\label{sec4}
In recent years, there have been different approaches to ``quantum Painlev\'e equations'' using as starting point the linear differential equation of rank~$2$ classically associated to the Painlev\'e equations, or the theory of topological recursion associated to semiclassical spectral curves.

In \cite{ZZ}, the authors start from the system of linear equations (Lax system) associated to the Painlev\'e equations {I}--{VI} written in the form
\begin{equation}
\begin{cases}
\partial_{z}\Psi= U(z,t)\Psi,\\
\partial_{t}\Psi= V(z,t)\Psi.
\end{cases}
\label{laxpp}
\end{equation}
After applying a suitable change of variables and gauge transformation, they obtain a pair of compatible PDEs for a scalar wave function $\psi$ (obtained from the $(1,1)$ entry of the matrix $\Psi$)
\begin{equation}
\begin{cases}
\big(\frac{1}{2}\partial_z^2 -\frac{1}{2}(\partial_z \log b)\partial_z+W(z,t)\big)\psi=0, \\
\partial_t\psi=\big(\frac{1}{2}\partial_z^2+\mathcal U(z,t)\big)\psi,
\end{cases}\label{ZZ-lds}
\end{equation}
where $W$ and $\mathcal U$ are described explicitly in terms of the entries of the matrices $U$,~$V$.
The first equation in \eqref{ZZ-lds} has apparent singularities but otherwise exhibits the same (generalized) mo\-no\-dromy associated to \eqref{laxpp} (in $\mathbb P {\rm SL}_2$), while the second equation describes the isomonodromic deformation of the former and is presented in the form of a non-stationary Schr\"odinger equation with imaginary time.

The Hamiltonian operators corresponding to each Painlev\'e equation {I}--{VI} from the second equation of the system~\eqref{ZZ-lds} are, in particular, a natural quantization of those corresponding to the Calogero-like Painlev\'e equations (those obtained from the first equation of the system~\eqref{ZZ-lds}). These operators that are called the quantum Calogero--Painlev\'e Hamiltonian system, are obtained in single variable representation.
The question then arises as to whether a similar description is possible in the multi-particle case; the na\"ive approach of considering the matrices as $2\times 2$ blocks does not lead to equations of the same type as \eqref{ZZ-lds}. The quantization of the Hamiltonians from~\cite{BerCafRub} does not seem to be the direct analog of~\eqref{ZZ-lds} because it is an equation where the wave equation plays the role rather of the (scalar) ``quantum tau function''.

Also, in \cite{Eyn}, the authors use the topological recursion on spectral curves of different genuses which results in wave functions that satisfy a family of partial differential equations. In fact, these PDEs are the quantization of the original spectral curves. As an application of their theorem, they introduce a system of PDEs corresponding to Painlev\'e transcendents whose assigned Hamiltonian system has significant similarities to the system of Hamiltonian operators that we introduced in this paper for the quantum Calogero--Painlev\'e system.

Finally, we comment on the possible relationship with equations of the Knizhnik--Zamo\-lod\-chi\-kov (KZ) type. For the single-particle case, in \cite{Nag}, H.~Nagoya provides a representation-theoretic correspondence between the Schr\"odinger equation for quantum Painlev\'e~{VI} (single-particle) and the Knizhnik--Zamolodchikov (KZ) equation, and between the Schr\"odinger equation for quantum Painlev\'e {II--V} (single-particle) and the confluent KZ equations that are defined in~\cite{JimNag}. These correspondences are proved directly by showing relations between the integral representations for the solutions to the quantum Painlev\'e equations and solutions to the (confluent) KZ equations. In this case, the present paper should play a similar role for $\mathfrak{sl}_n$ special solutions of the KZ equations.
We plan to address this possible relationship in a subsequent publication.

\appendix
\section{Proofs}

The starting point of all the following computations is to take the integral representation \eqref{phiNp}
\begin{equation*}
\Phi(\vec z;t)=\int_{\Gamma} \prod_{1\leq i<j\leq m}(u_i-u_j)^{2\hbar} \prod_{\rho,i}(z_{\rho}-u_i)\Theta_J(u_i) \d u_i,
\end{equation*}
and apply to it the direct sum of the second-order parts in the quantum Hamiltonians \eqref{NagHamil}.

For the sake of simplicity of the notation in the following computations, We denote by $\langle F(u_1,\dots ,u_m)\rangle$ the un-normalized expectation value as follows
\begin{equation*}
\langle F(u_1,\dots ,u_m)\rangle :=\int_{\Gamma} \prod_{1\leq i<j\leq m}(u_i-u_j)^{2\hbar} \prod_{i}\Theta_J(u_i)F(u_1,\dots ,u_m)\d u_i.
\end{equation*}

\subsection{Quantum Painlev\'e II}

Define $P(\vec z)=\prod_{\rho,i} ({ z}_{\rho}-u_i)$ and $\Delta=\prod_{1\leq i<j\leq m}(u_i-u_j)$, then
\begin{gather*}
\sum_{\rho}\partial_{{ z}_{\rho}}^2 P=
\sum_{\rho} P \sum_{i\neq j}\frac{1}{({ z}_{\rho}-u_i)({ z}_{\rho}-u_j)} \\
\hphantom{\sum_{\rho}\partial_{{ z}_{\rho}}^2 P}{}
=
\sum_{\rho}P\sum_{i\neq j}\left(\frac{1}{({ z}_{\rho}-u_i)(u_i-u_j)}-\frac{1}{({ z}_{\rho}-u_j)(u_i-u_j)}\right).
\end{gather*}
This yields
\begin{gather*}
\hbar^2 \sum_{\rho}\partial_{{ z}_{\rho}}^2\Phi
= \hbar^2\bigg\langle P\sum_{\rho}\sum_{i\neq j}\left(\frac{1}{({ z}_{\rho}-u_i)(u_i-u_j)}-\frac{1}{({ z}_{\rho}-u_j)(u_i-u_j)}\right)\bigg\rangle \\
\hphantom{\hbar^2 \sum_{\rho}\partial_{{ z}_{\rho}}^2\Phi}{}
=
2\hbar^2\bigg\langle P\sum_{\rho}\sum_{i\neq j}\frac{1}{({ z}_{\rho}-u_i)(u_i-u_j)}\bigg\rangle
\\
\hphantom{\hbar^2 \sum_{\rho}\partial_{{ z}_{\rho}}^2\Phi}{}
= \hbar \int\sum_{i} \partial_{u_i}\big(\Delta^{2\hbar}\big)\sum_{\rho}\frac{1}{{ z}_{\rho}-u_i}P\prod_{k}\Theta(u_k) \d u_i.
\end{gather*}
We now use integration by parts in the integrand and obtain
\begin{gather*}
= -\hbar\int \Delta^{2\hbar}\sum_{\rho}\sum_{i}\partial_{u_i}\bigg(\frac{1}{{ z}_{\rho}-u_i}P\prod_{k}\Theta(u_k)\bigg)\d u_i \\
=
-\hbar \int \Delta^{2\hbar}\sum_{\rho}\sum_{i} \left(\frac{P}{({ z}_\rho-u_i)^2}+\frac{P_{u_i}}{{ z}_\rho-u_i}-\frac{2u_i^2+t}{{ z}_\rho-u_i}P\right)\prod_{k}\Theta(u_k)\d u_i.
\end{gather*}
After simplifying we obtain
\begin{gather*}
\frac{\hbar^2}{2} \sum_{\rho}\partial_{{ z}_{\rho}}^2\Phi
=
 -\frac{\hbar}{2}\underset{\rho\neq\sigma}{\sum_{\rho,\sigma}} \frac{\partial_{{ z}_{\rho}}-\partial_{{ z}_{\sigma}}}{{ z}_{\rho}-{ z}_{\sigma}} \Phi + \hbar\sum_{\rho}\left({ z}_{\rho}^2+\frac{t}{2}\right)\partial_{{ z}_{\rho}}\Phi- \hbar m \sum_{\rho}{ z}_{\rho}\Phi +\hbar N \partial_t \Phi .
\end{gather*}
Rearranging the terms appropriately, we obtain the Schr\"odinger equation with the Hamiltonian~\eqref{CP2}:
\begin{gather*}
H_{\rm II}= \frac{\hbar}{2}\underset{\rho\neq\sigma}{\sum_{\rho,\sigma}} \frac{\partial_{{ z}_{\rho}}-\partial_{{ z}_{\sigma}}}{{ z}_{\rho}-{ z}_{\sigma}} +\frac{\hbar^2}{2} \sum_{\rho}\partial_{{ z}_{\rho}}^2 -\hbar\sum_{\rho} \left({ z}_{\rho}^2+\frac{t}{2}\right)\partial_{{ z}_{\rho}} +m\hbar \sum_{\rho} { z}_{\rho} .
\end{gather*}

\subsection{Quantum Painlev\'e III}
The initial setup matches the previous case, except that we need to consider the direct sum of the operators ${ z}_\rho^2 \pa_{{ z}_\rho}^2$ acting on the wave function. To this end, we observe that
\begin{gather*}
\sum_{\rho} { z}_{\rho}^2 \partial_{{ z}_{\rho}}^2 P({\vec z})
=
\sum_{\rho} P \sum_{i\neq j}\frac{{ z}_{\rho}^2}{({ z}_{\rho}-u_i)({ z}_{\rho}-u_j)} \\
\hphantom{\sum_{\rho} { z}_{\rho}^2 \partial_{{ z}_{\rho}}^2 P({\vec z})}{}
= \sum_{\rho} P \sum_{i\neq j} \left( 1+\frac{u_i^2}{({ z}_{\rho}-u_i)(u_i-u_j)}- \frac{u_j^2}{({ z}_{\rho}-u_j)(u_i-u_j)} \right).
\end{gather*}
This yields
\begin{gather*}
\hbar^2\sum_{\rho}{ z}_{\rho}^2\partial_{{ z}_{\rho}}^2 \Phi
= \hbar^2\bigg\langle P\sum_{\rho}\sum_{i\neq j}\left( 1+\frac{u_i^2}{({ z}_{\rho}-u_i)(u_i-u_j)}- \frac{u_j^2}{({ z}_{\rho}-u_j)(u_i-u_j)} \right)\bigg\rangle \\
\hphantom{\hbar^2\sum_{\rho}{ z}_{\rho}^2\partial_{{ z}_{\rho}}^2 \Phi}{}
=
\hbar^2Nm(m-1)\Phi+ \hbar Nm\Phi+ \hbar N^2 m \Phi- \hbar N\sum_{\sigma}{ z}_{\sigma}\partial_{{ z}_{\sigma}}\Phi- \hbar \sum_{\rho}{ z}_{\rho}\sum_{\sigma}\partial_{{ z}_{\sigma}}\Phi
\\
\hphantom{\hbar^2\sum_{\rho}{ z}_{\rho}^2\partial_{{ z}_{\rho}}^2 \Phi=}{} +\hbar \sum_{\rho}({ z}_{\rho}^2+b{ z}_{\rho}+t)\partial_{{ z}_{\rho}}\Phi- \hbar bNm\Phi- \hbar Nm\Phi+ \hbar \sum_{\rho}\partial_{{ z}_{\rho}}\Phi- \hbar m \sum_{\rho}{ z}_{\rho}\Phi
\\
\hphantom{\hbar^2\sum_{\rho}{ z}_{\rho}^2\partial_{{ z}_{\rho}}^2 \Phi=}{}
-\hbar N\bigg\langle\! \sum_{i}u_i P\!\bigg\rangle -\hbar \underset{\rho\neq\sigma}{\sum_{\rho,\sigma}}\frac{{ z}_{\rho}^2\partial_{{ z}_{\rho}}-{ z}_{\sigma}^2\partial_{{ z}_{\sigma}}}{{ z}_{\rho}-{ z}_{\sigma}}\Phi +\hbar \sum_{\rho}{ z}_{\rho}\sum_{\sigma}\partial_{{ z}_{\sigma}}\Phi- \hbar \sum_{\rho}{ z}_{\rho}\partial_{{ z}_{\rho}}\Phi
\\
\hphantom{\hbar^2\sum_{\rho}{ z}_{\rho}^2\partial_{{ z}_{\rho}}^2 \Phi=}{}
+\hbar (N-1)\sum_{\sigma}{ z}_{\sigma}\partial_{{ z}_{\sigma}}\Phi.
\end{gather*}
Rearranging the terms we obtain the following expression
\begin{gather}
\hbar^2\sum_{\rho}{ z}_{\rho}^2\partial_{{ z}_{\rho}}^2 \Phi
=
-\hbar\underset{\rho\neq\sigma}{\sum_{\rho,\sigma}}\frac{{ z}_{\rho}^2\partial_{{ z}_{\rho}}-{ z}_{\sigma}^2\partial_{{ z}_{\sigma}}}{{ z}_{\rho}-{ z}_{\sigma}}\Phi+\hbar \sum_{\rho}\big({ z}_{\rho}^2+b{ z}_{\rho}+t\big)\partial_{{ z}_{\rho}}\Phi -m\hbar \sum_{\rho}{ z}_{\rho}\Phi \nn
\\
\hphantom{\hbar^2\sum_{\rho}{ z}_{\rho}^2\partial_{{ z}_{\rho}}^2 \Phi=}{}
+ t\hbar N \partial_{t}\Phi+\hbar^2Nm(m-1)\Phi+ \hbar N^2m\Phi- bNm\hbar\Phi- \hbar \sum_{\rho}{ z}_{\rho}\partial_{{ z}_{\rho}}\Phi \nn \\
\hphantom{\hbar^2\sum_{\rho}{ z}_{\rho}^2\partial_{{ z}_{\rho}}^2 \Phi=}{}
- \hbar N\bigg\langle \sum_i \left(\frac{t}{u_i}+u_i\right)P\bigg\rangle.\label{p3euler}
\end{gather}
In order to handle the remaining expectation value we need to derive some further identities: consider the Euler differential operator
\begin{equation*}
\mathbb E:= \sum_{i} u_i \partial_{u_i}.
\end{equation*}
Applying this operator $\mathbb E$ to the integrand of $\Phi_m^{\rm III}$,
\begin{equation*}
\prod_{1\leq i<j \leq m}(u_i-u_j)^{2\hbar} \prod_{\rho,i}({ z}_{\rho}-u_i)u_i^{-b-1}{\rm e}^{(\frac{t}{u_i}-u_i)}
\end{equation*}
gives the following expression
\begin{gather*}
\sum_{i} u_i \partial_{u_i}\bigg(
\prod_{1\leq i<j \leq m}(u_i-u_j)^{2\hbar}
\prod_{\rho,i}({ z}_{\rho}-u_i)u_i^{-b-1}{\rm e}^{(\frac{t}{u_i}-u_i)}\bigg) \nn
\\
=\sum_{i} u_i \bigg[\prod_{i<j}(u_i-u_j)^{2\hbar} \bigg((2\hbar)\sum_{i\neq j}\frac{1}{u_i-u_j}\bigg) \prod_{\rho,i}({ z}_{\rho}-u_i)u_i^{-b-1}{\rm e}^{(\frac{t}{u_i}-u_i)} \nn
\\
\quad{} +\prod_{i<j}(u_i-u_j)^{2\hbar} \prod_{\rho,i}({ z}_{\rho}-u_i)\bigg(\sum_{\rho}\frac{-1}{{ z}_{\rho}-u_i}\bigg)u_i^{-b-1}{\rm e}^{(\frac{t}{u_i}-u_i)} \nn
\\
\quad{} +\prod_{i<j}(u_i-u_j)^{2\hbar} \prod_{\rho,i}({ z}_{\rho}-u_i)u_i^{-b-1}{\rm e}^{(\frac{t}{u_i}-u_i)}\left(\frac{-b-1}{u_i}-\frac{t}{u_i^2}-1\right)\bigg]
\\
=
\sum_i u_i\bigg[2\hbar \sum_{i\neq j}\frac{1}{u_i-u_j}-\sum_{\rho}\frac{1}{{ z}_{\rho}-u_i}-\frac{b+1}{u_i}-\frac{t}{u_i^2}-1\bigg]\Delta^{2\hbar}P\prod_{i}\Theta(u_i)
\\
= \Delta^{2\hbar}\bigg[\hbar m(m-1)P+ NmP- \sum_{\rho}{z}_{\rho}\partial_{z_{\rho}}P- (b+1)mP- \sum_i\left(\frac{t}{u_i}+u_i\right)P\bigg] \prod_{i}\Theta(u_i) .
\end{gather*}
Then we observe that
\begin{gather*}
\int \sum_j u_j \partial_{u_j}\bigg(\Delta^{2\hbar}P\prod_{i}\Theta(u_i)\d u_i\bigg) =
\int \sum_j \partial_{u_j}\bigg(u_j \Delta^{2\hbar}P\prod_{i}\Theta(u_i)\d u_i\bigg) - m\Phi ( z) ,
\end{gather*}
and the first integral is zero because the integrand is a divergence of a vector.
Therefore,
\begin{gather*}
\underbrace{
\int\mathbb E \bigg(\Delta^{2\hbar}P\prod_{i}\Theta(u_i) \d u_i\bigg)}_{=-m\Phi } \\
\qquad{} =
\hbar m(m-1)\Phi + Nm\Phi - \sum_{\rho}{ z}_{\rho}\partial_{{ z}_{\rho}}\Phi
-(b+1)m\Phi
-
 \bigg\langle\sum_i \left(\frac{t}{u_i}+u_i\right)P\bigg\rangle\\
\Longrightarrow
\quad
\bigg\langle \sum_i (\frac{t}{u_i}+u_i)P\bigg\rangle = \hbar m(m-1)\Phi + Nm\Phi - \sum_{\rho}{ z}_{\rho}\partial_{{ z}_{\rho}}\Phi -b\Phi .
\end{gather*}
Substituting the left side into the equation \eqref{p3euler} results in the following conclusion
\begin{gather*}
t\hbar N\partial_{t}\Phi =
\sum_{\rho}\big(\hbar^2 { z}_{\rho}^2\partial_{z_{\rho}}^2-\hbar \big({ z}_{\rho}^2+(b+N-1){ z}_{\rho}+t\big)\partial_{{ z}_{\rho}}+m\hbar { z}_{\rho}\big)\Phi +\hbar\underset{\rho\neq\sigma}{\sum_{\rho,\sigma}}\frac{{ z}_{\rho}^2\partial_{{ z}_{\rho}}-{ z}_{\sigma}^2\partial_{{ z}_{\sigma}}}{{ z}_{\rho}-{ z}_{\sigma}}\Phi.
\end{gather*}
Hence, the general Hamiltonian operator for quantum Painlev\'e {III} equation is given by
\begin{gather*}
tH_{\rm III}= \hbar\underset{\rho\neq\sigma}{\sum_{\rho,\sigma}}\frac{{ z}_{\rho}^2\partial_{{ z}_{\rho}}-{ z}_{\sigma}^2\partial_{{ z}_{\sigma}}}{{ z}_{\rho}-{ z}_{\sigma}}+ \sum_{\rho}\big(\hbar^2 { z}_{\rho}^2\partial_{{ z}_{\rho}}^2-\hbar \big({ z}_{\rho}^2+(b+N-1){ z}_{\rho}+t\big)\partial_{{ z}_{\rho}}+m\hbar { z}_{\rho}\big).
\end{gather*}

\subsection{Quantum Painlev\'e IV}
We follow the same general scheme and consider the sum of the $z_\rho\partial_{z\rho}^2$ term applied to $P$:
\begin{gather*}
\sum_{\rho}{ z}_{\rho}\partial_{{ z}_{\rho}}^2 P=
\sum_{\rho}{ z}_{\rho} P \sum_{i\neq j}\frac{1}{({ z}_{\rho}-u_i)({ z}_{\rho}-u_j)} \nn \\
\hphantom{\sum_{\rho}{ z}_{\rho}\partial_{{ z}_{\rho}}^2 P}{}
=
\sum_{\rho}P\sum_{i\neq j}\left(\frac{u_i}{({ z}_{\rho}-u_i)(u_i-u_j)}-\frac{u_j}{({ z}_{\rho}-u_j)(u_i-u_j)}\right).
\end{gather*}
This yields
\begin{gather*}
\hbar^2 \sum_{\rho}{ z}_{\rho}\partial_{{ z}_{\rho}}^2\Phi
= \hbar^2\bigg\langle P\sum_{\rho}\sum_{i\neq j}\left(\frac{u_i}{({ z}_{\rho}-u_i)(u_i-u_j)}-\frac{u_j}{({ z}_{\rho}-u_j)(u_i-u_j)}\right)\bigg\rangle
\\
\hphantom{\hbar^2 \sum_{\rho}{ z}_{\rho}\partial_{{ z}_{\rho}}^2\Phi}{} =
2\hbar^2\bigg\langle P\sum_{\rho}\sum_{i\neq j}\frac{u_i}{({ z}_{\rho}-u_i)(u_i-u_j)}\bigg\rangle
\\
\hphantom{\hbar^2 \sum_{\rho}{ z}_{\rho}\partial_{{ z}_{\rho}}^2\Phi}{} = \hbar \int\sum_{i} \partial_{u_i}\big(\Delta^{2\hbar}\big)\sum_{\rho}\frac{u_i}{{ z}_{\rho}-u_i}P\prod_{k}\Theta(u_k)\d u_i
\\
\hphantom{\hbar^2 \sum_{\rho}{ z}_{\rho}\partial_{{ z}_{\rho}}^2\Phi}{} =
-\hbar\int \Delta^{2\hbar}\sum_{\rho}\sum_{i}\partial_{u_i}\bigg(\left(-1+\frac{{ z}_{\rho}}{{ z}_{\rho}-u_i}\right)P\prod_{k}\Theta(u_k)\bigg)\d u_i,
\end{gather*}
which gives
\begin{gather*}
 \hbar^2 \sum_{\rho}{ z}_{\rho}\partial_{{ z}_{\rho}}^2\Phi
=
 -\hbar\underset{\rho\neq\sigma}{\sum_{\rho,\sigma}} \frac{{ z}_{\rho}\partial_{{ z}_{\rho}}-{ z}_{\sigma}\partial_{{ z}_{\sigma}}}{{ z}_{\rho}-{ z}_{\sigma}} \Phi + \hbar\sum_{\rho}\big({ z}_{\rho}^2+t{ z}_{\rho}+b\big)\partial_{{ z}_{\rho}}\Phi \\
\hphantom{\hbar^2 \sum_{\rho}{ z}_{\rho}\partial_{{ z}_{\rho}}^2\Phi=}{} -\hbar Nm t \Phi + \hbar N\partial_{t}\Phi -\hbar m\sum_{\rho}{ z}_{\rho}\Phi .
\end{gather*}
Rearranging the terms in the above expression, the general Hamiltonian operator for quantum Painlev\'e {IV} equation is given by
\begin{gather*}
H_{\rm IV}= \hbar\underset{\rho\neq\sigma}{\sum_{\rho,\sigma}}\frac{{ z}_{\rho}\partial_{{ z}_{\rho}}-{ z}_{\sigma}\partial_{{ z}_{\sigma}}}{{ z}_{\rho}-{ z}_{\sigma}}+ \sum_{\rho}\big(\hbar^2 { z}_{\rho}\partial_{{ z}_{\rho}}^2-\hbar \big({ z}_{\rho}^2+t{ z}_{\rho}+b\big)\partial_{{ z}_{\rho}}+m\hbar { z}_{\rho}\big) + \hbar Nmt .
\end{gather*}

\subsection{Quantum Painlev\'e V}
We start from the same set up as previous operators
\begin{gather*}
\sum_{\rho} { z}_{\rho}({ z}_{\rho}-1)\partial_{{ z}_{\rho}}^2P
=
\sum_{\rho}P \sum_{i\neq j}\frac{{ z}_{\rho}({ z}_{\rho}-1)}{({ z}_{\rho}-u_i)({ z}_{\rho}-u_j)} \\
\hphantom{\sum_{\rho} { z}_{\rho}({ z}_{\rho}-1)\partial_{{ z}_{\rho}}^2P}{}
=
\sum_{\rho}P \sum_{i\neq j} \left(1+ \frac{u_i(u_i-1)}{({ z}_{\rho}-u_i)(u_i-u_j)}-\frac{u_j(u_j-1)}{({ z}_{\rho}-u_j)(u_i-u_j)}\right).
\end{gather*}
This yields
\begin{gather*}
\hbar^2\sum_{\rho} { z}_{\rho}({ z}_{\rho}-1)\partial_{{ z}_{\rho}}^2\Phi =
\hbar^2\bigg\langle P\sum_{\rho}\sum_{i\neq j}\left(1+ \frac{u_i(u_i-1)}{({ z}_{\rho}-u_i)(u_i-u_j)}-\frac{u_j(u_j-1)}{({ z}_{\rho}-u_j)(u_i-u_j)}\right)\bigg\rangle
\\
\hphantom{\hbar^2\sum_{\rho} { z}_{\rho}({ z}_{\rho}-1)\partial_{{ z}_{\rho}}^2\Phi}{}
=\hbar^2Nm(m-1)\Phi +2\hbar^2\bigg \langle P\sum_{\rho}\sum_{i\neq j}\frac{u_i(u_i-1)}{({ z}_{\rho}-u_i)(u_i-u_j)}\bigg\rangle
\\
\hphantom{\hbar^2\sum_{\rho} { z}_{\rho}({ z}_{\rho}-1)\partial_{{ z}_{\rho}}^2\Phi}{}
=\hbar^2Nm(m-1)\Phi +\hbar\int\sum_{i}\partial_{u_i}\big(\Delta^{2\hbar}\big)\sum_{\rho}\frac{u_i(u_i-1)}{{ z}_{\rho}-u_i}P\prod_{k}\Theta(u_k)\d u_i,
\\
\hbar^2\sum_{\rho} { z}_{\rho}({ z}_{\rho}-1)\partial_{{ z}_{\rho}}^2\Phi =
-\hbar\underset{\rho\neq\sigma}{\sum_{\rho,\sigma}}\frac{{ z}_{\rho}({ z}_{\rho}-1)\partial_{{ z}_{\rho}}-{ z}_{\sigma}({ z}_{\sigma}-1)\partial_{{ z}_{\sigma}}}{{ z}_{\rho}-{ z}_{\sigma}}\Phi \nn
\\
\hphantom{\hbar^2\sum_{\rho} { z}_{\rho}({ z}_{\rho}-1)\partial_{{ z}_{\rho}}^2\Phi =}{}
 -\hbar\sum_{\rho}\big(t{ z}_{\rho}^2-(b+c+t){ z}_{\rho}+b\big)\partial_{{ z}_{\rho}}\Phi -\hbar Nm(b+c+t)\Phi \nn \\
\hphantom{\hbar^2\sum_{\rho} { z}_{\rho}({ z}_{\rho}-1)\partial_{{ z}_{\rho}}^2\Phi =}{}
 + m\hbar t \sum_{\rho}{ z}_{\rho}\Phi + \hbar N t\partial_t \Phi + \hbar^2 Nm(m-1)\Phi + \hbar Nm(N-1)\Phi.
\end{gather*}
Therefore, the general Hamiltonian operator for quantum Painlev\'e~{V} equation is given by
\begin{gather*}
tH_{\rm V}=
\hbar\underset{\rho\neq\sigma}{\sum_{\rho,\sigma}}\frac{{ z}_{\rho}({ z}_{\rho}-1)\partial_{{ z}_{\rho}}-{ z}_{\sigma}({ z}_{\sigma}-1)\partial_{{ z}_{\sigma}}}{{ z}_{\rho}-{ z}_{\sigma}}\\
\hphantom{tH_{\rm V}=}{}
+ \sum_{\rho}\big(\hbar^2 { z}_{\rho}({ z}_{\rho}-1)\partial_{{ z}_{\rho}}^2+\hbar \big(t{ z}_{\rho}^2-(b+c+t){ z}_{\rho}+b\big)\partial_{{ z}_{\rho}}
- m\hbar t{ z}_{\rho}\big)\\
\hphantom{tH_{\rm V}=}{}
+\hbar Nm(b+c+t)-\hbar^2Nm(m-1)-\hbar Nm(N-1).
\end{gather*}

\subsection{Quantum Painlev\'e VI}
The Laplacian operator applied to $P$ yields
\begin{gather*}
\sum_{\rho}{ z}_{\rho}({ z}_{\rho}-1)({ z}_{\rho}-t)\partial_{{ z}_{\rho}}^2 P
=\sum_{\rho}{ z}_{\rho}({ z}_{\rho}-1)({ z}_{\rho}-t)P\sum_{i\neq j}\frac{1}{({ z}_{\rho}-u_i)({ z}_{\rho}-u_j)} \\
\qquad =P\sum_{\rho}\sum_{i\neq j}\frac{{ z}_{\rho}({ z}_{\rho}-1)({ z}_{\rho}-t)}{({ z}_{\rho}-u_i)({ z}_{\rho}-u_j)} \\
\qquad =P\sum_{\rho}\sum_{i\neq j} \left( ({ z}_{\rho}-(t+1)+u_i+u_j )+\frac{u_i(u_i-1)(u_i-t)}{(u_i-u_j)({ z}_{\rho}-u_i)}
 -\frac{u_j(u_j-1)(u_j-t)}{(u_i-u_j)({ z}_{\rho}-u_j)}\right).
\end{gather*}
This yields
\begin{gather*}
\hbar^2\sum_{\rho}{ z}_{\rho}({ z}_{\rho}-1)({ z}_{\rho}-t)\partial_{{ z}_{\rho}}^2 \Phi\\
=
\hbar^2 \bigg\langle P\sum_{\rho}\sum_{i\neq j} \left( ({ z}_{\rho}-(t+1)+u_i+u_j )+\frac{u_i(u_i-1)(u_i-t)}{(u_i-u_j)({ z}_{\rho}-u_i)}
-\frac{u_j(u_j-1)(u_j-t)}{(u_i-u_j)({ z}_{\rho}-u_j)}\right)\bigg\rangle.
\end{gather*}
Upon rearranging of the terms we obtain the equation
\begin{gather}
\hbar^2\sum_{\rho}{ z}_{\rho}({ z}_{\rho}-1) ({ z}_{\rho}-t)\partial_{{ z}_{\rho}}^2\Phi =
-\hbar\underset{\rho\neq\sigma}{\sum_{\rho,\sigma}}\frac{{ z}_{\rho}({ z}_{\rho}-1)({ z}_{\rho}-t)\partial_{{ z}_{\rho}}-{ z}_{\sigma}({ z}_{\sigma}-1)({ z}_{\sigma}-t)\partial_{{ z}_{\sigma}}}{{ z}_{\rho}-{ z}_{\sigma}} \Phi \nn \\
\qquad{} +t(t-1)\hbar\partial_t \Phi
+\hbar\sum_{\rho}\big((a+b)({ z}_{\rho}-1)({ z}_{\rho}-t)+c{ z}_{\rho}({ z}_{\rho}-t)+d{ z}_{\rho}({ z}_{\rho}-1)\big) \partial_{{ z}_{\rho}}\Phi \nn \\
\qquad{} -\hbar\sum_{\rho}{ z}_{\rho}^2\partial_{{ z}_{\rho}}\Phi + \hbar\sum_{\rho}{ z}_{\rho}\partial_{{ z}_{\rho}}\Phi + \big(2\hbar Nm-\hbar m- \hbar^2 m^2\big)\sum_{\rho}{ z}_{\rho}\Phi \nn \\
\qquad{} +\big(\hbar^2 Nm-\hbar tN^2m-\hbar N^2m+m^2\hbar^2Nt-d\hbar Nmt+\hbar Nmt+(b+d)\hbar Nm\big)\Phi \nn \\
\qquad{} +(\hbar N^2-\hbar^2 N)\bigg\langle \sum_i u_i P\bigg\rangle + t(t-1)\hbar N\bigg\langle \sum_i \frac{d}{t-u_i}P\bigg\rangle .\label{p6op}
\end{gather}
Now, instead of applying the Euler operator, we consider the operator
\begin{equation*}
{\mathbb L}:=\sum_i \partial_{u_i} (u_i(1-u_i) ).
\end{equation*}
We apply $\mathbb L$ to the integrand
\begin{equation*}
J( u, z):=
\prod_{1\leq i<j\leq m}(u_i-u_j)^{2\hbar} \prod_{\rho, i} ({ z}_{\rho}-u_i)u_i^{-a-b-1}(1-u_i)^{-c-1}(t-u_i)^{-d}.
\end{equation*}
After some calculations we obtain
\begin{gather}
\mathbb L J( \vec u, \vec z)=
\sum_i \bigg[(1-2u_i)\prod_{i<j}(u_i-u_j)^{2\hbar} \prod_{\rho, i} ({ z}_{\rho}-u_i)u_i^{-a-b-1}(1-u_i)^{-c-1}(t-u_i)^{-d} \nn \\
{}+u_i(1-u_i)\prod_{i<j}(u_i-u_j)^{2\hbar}\bigg(2\hbar \sum_{i\neq j}\frac{1}{u_i-u_j}\bigg) \prod_{\rho, i} ({ z}_{\rho}-u_i)u_i^{-a-b-1}(1-u_i)^{-c-1}(t-u_i)^{-d} \nn \\
{} +u_i(1-u_i)\prod_{i<j}(u_i-u_j)^{2\hbar} \prod_{\rho, i} ({ z}_{\rho}-u_i)\bigg(\sum_{\rho}\frac{-1}{{ z}_{\rho}-u_i}\bigg)\prod_{i}u_i^{-a-b-1}(1-u_i)^{-c-1}(t-u_i)^{-d} \nn \\
{}+u_i(1-u_i)\prod_{i<j}(u_i-u_j)^{2\hbar} \prod_{\rho, i} \frac{({ z}_{\rho}-u_i)}
{u_i^{a+b+1}(1-u_i)^{c+1}(t-u_i)^{d}}\nonumber\\
\quad{}\times \left(\frac{-a-b-1}{u_i}+\frac{c+1}{1-u_i}+\frac{d}{t-u_i}\right)
\bigg]\nonumber
\\
{} =
\Delta^{2\hbar}\bigg[mP-2\sum_i u_i P-2\hbar(m-1)\sum_i u_i P- m\sum_{\rho}{ z}_{\rho}P-N\sum_i u_i P+ NmP \nn \\
{} -\sum_{\rho}{ z}_{\rho}(1-{ z}_{\rho})\partial_{{ z}_{\rho}}P+(-a-b-d-1)mP+tdmP \nonumber\\
{} +(a+b+c+d+2)\sum_i u_i P -t(t-1)\sum_i \frac{d}{t-u_i}P\bigg]\prod_i \Theta(u_i) .\label{E3}
\end{gather}
Integrating \eqref{E3}, one obtains zero because the integrand $\mathbb L J$ can be viewed as the divergence of a vector field. Therefore,
\begin{gather*}
t(t-1)\bigg\langle \sum_i \frac{d}{t-u_i}P\bigg\rangle =
(-\hbar m+Nm+(-b-d)m+tdm)\Phi \nn \\
\hphantom{t(t-1)\bigg\langle \sum_i \frac{d}{t-u_i}P\bigg\rangle =}{}
+ (-2\hbar(m-1)-N+a+b+c+d)\bigg\langle \sum_i u_i P\bigg\rangle \\
\hphantom{t(t-1)\bigg\langle \sum_i \frac{d}{t-u_i}P\bigg\rangle =}{}
-m\sum_{\rho}{ z}_{\rho}\Phi - \sum_{\rho}{ z}_{\rho}(1-{ z}_{\rho})\partial_{{ z}_{\rho}}\Phi.
\end{gather*}
Substituting the l.h.s.\ into the equation \eqref{p6op} yields the following result
\begin{gather*}
\hbar^2\sum_{\rho}{ z}_{\rho}({ z}_{\rho}-1)({ z}_{\rho}-t)\partial_{{ z}_{\rho}}^2\Phi =
-\hbar \underset{\rho\neq\sigma}{\sum_{\rho,\sigma}}\frac{{ z}_{\rho}({ z}_{\rho}-1)({ z}_{\rho}-t)\partial_{{ z}_{\rho}}-{ z}_{\sigma}({ z}_{\sigma}-1)({ z}_{\sigma}-t)\partial_{{ z}_{\sigma}}}{{ z}_{\rho}-{ z}_{\sigma}} \Phi \\
\qquad{} +t(t-1)\hbar\partial_t \Phi+\hbar \sum_{\rho}\big((a+b)({ z}_{\rho}-1)({ z}_{\rho}-t)+c{ z}_{\rho}({ z}_{\rho}-t)+d{ z}_{\rho}({ z}_{\rho}-1)\big)\partial_{{ z}_{\rho}}\Phi \\
\qquad{} +\hbar (1-N)\sum_{\rho}{ z}_{\rho}(1-{ z}_{\rho})\partial_{{ z}_{\rho}}\Phi + \hbar m(N-1-\hbar m)\sum_{\rho}{ z}_{\rho}\Phi
+ \hbar m N(\hbar m+1-N)t\Phi.
\end{gather*}
Therefore, the general Hamiltonian operator for quantum Painlev\'e {VI} equation is given by
\begin{gather*}
t(t-1)H_{\rm VI}=
 \hbar\underset{\rho\neq\sigma}{\sum_{\rho,\sigma}}\frac{{ z}_{\rho}({ z}_{\rho}-1)({ z}_{\rho}-t)\partial_{{ z}_{\rho}}-{ z}_{\sigma}({ z}_{\sigma}-1)({ z}_{\sigma}-t)\partial_{{ z}_{\sigma}}}{{ z}_{\rho}-{ z}_{\sigma}}\\
 \hphantom{t(t-1)H_{\rm VI}=}{}
+ \sum_{\rho}\hbar^2 { z}_{\rho}({ z}_{\rho}-1)({ z}_{\rho}-t)\partial_{{ z}_{\rho}}^2 \\
\hphantom{t(t-1)H_{\rm VI}=}{}
-\hbar \sum_{\rho}\big((a+b)({ z}_{\rho}-1)({ z}_{\rho}-t)+c{ z}_{\rho}({ z}_{\rho}-t)+(d+N-1){ z}_{\rho}({ z}_{\rho}-1)\big)\partial_{{ z}_{\rho}}
\\
\hphantom{t(t-1)H_{\rm VI}=}{}
- \hbar m(N-1-\hbar m)\sum_{\rho}{ z}_{\rho}- \hbar m N(\hbar m+1-N)t.
\end{gather*}

\begin{Remark}
All these operators reduce to the Hamiltonian operators in \eqref{NagHamil} for $N=1$.
\end{Remark}

\subsection*{Acknowledgements}
The work of F.M.\ and M.B.\ was supported in part by the Natural Sciences and Engineering Research Council of Canada (NSERC) grant RGPIN-2016-06660.

\pdfbookmark[1]{References}{ref}
\LastPageEnding

\end{document}